\documentclass[12pt]{article}

\usepackage{amsmath,amssymb,amsthm,appendix,bm,booktabs,caption,graphicx,multirow,natbib,physics}
\usepackage[margin=1in]{geometry}
\usepackage[singlespacing]{setspace}
\usepackage[dvipsnames]{xcolor}
\usepackage[colorlinks,citecolor=Blue,linkcolor=Blue,urlcolor=Maroon]{hyperref}

\DeclareMathOperator{\corr}{corr}
\DeclareMathOperator{\cov}{cov}
\DeclareMathOperator{\sd}{sd}
\DeclareMathOperator{\vvar}{var}

\newtheorem{definition}{Definition}
\newtheorem{proposition}{Proposition}

\title{Optimal macroprudential policy with preemptive bailouts\thanks{I would like to thank Toni Braun, Kaiji Chen, Tasos Karantounias, Federico Mandelman, Ricardo Nunes, Franck Portier, Juan Rubio-Ram\'{i}rez, Kjetil Storesletten, Vivian Yue, and seminar and conference participants at the Federal Reserve Bank of Atlanta, Emory University, Bank of Lithuania, University of Surrey, 2022 Econometric Society meetings, 2023 Bank of Portugal Workshop on Financial Stability and Macroprudential Policy, 2nd UCL-Surrey Workshop in Macroeconomics, 54th Annual Conference of the Money, Macro and Finance Society, 4th Macroeconomics Network in the Southwest Workshop, 2023 European Winter Meeting of the Econometric Society, 2024 RCEA International Conference, and 2024 European Summer Meeting of the Econometric Society for helpful comments and suggestions.}}
\author{Aliaksandr Zaretski\footnote{School of Social Sciences, University of Surrey, e-mail: \url{a.zaretski@surrey.ac.uk}.}}
\date{April 6, 2025}

\begin{document}

\maketitle

\begin{abstract}
    I study the optimal regulation of a financial sector where individual banks face self-enforcing constraints countering their default incentives. The constrained-efficient social planner can improve over the unregulated equilibrium in two dimensions. First, by internalizing the impact of banks' portfolio decisions on the prices of assets and liabilities that affect the enforcement constraints. Second, by redistributing future net worth from new entrants to surviving banks, which increases the current forward-looking value of all banks, relaxing their enforcement constraints and decreasing the probability of banking crises. The latter can be accomplished with systemic preemptive bailouts that are time consistent and unambiguously welfare improving. Unregulated banks can be both overleveraged and underleveraged depending on the state of the economy, thus macroprudential policy requires both taxes and subsidies, while minimum bank capital requirements are generally ineffective.
\end{abstract}


\newpage

\section{Introduction}
The global financial crisis and the Great Recession of the late 2000s raised several challenging normative questions. Is there a need for macroprudential regulation, and if so, which policy instruments are effective? Should regulatory requirements be conditioned on individual-specific characteristics? Is ``too big to fail'' a problem? Are bailouts justified? In the recent decade, considerable progress has been made in understanding the rationales for macroprudential regulation in small open economies that borrow in the international financial market at exogenously determined interest rates. At the same time, our knowledge of the optimal regulation of banks in quantitative general equilibrium remains limited. Banking crises were at the heart of the global financial crisis of 2007--2009, including in the US and the UK, and many developed and developing economies have a bank-based financial system. Therefore, it is crucial to understand how to regulate banks optimally over the business cycle.

This paper considers a quantitative general equilibrium environment with endogenous financial frictions in the banking system. In this environment, multiple externalities arise, justifying system-wide regulation, e.g., bank balance sheet taxation or minimum capital requirements. Without regulation, occasional large drops in net worth lead to financial crises when endogenous financial constraints switch from having slack to being binding. An alternative way to decrease the probability of financial crises and improve welfare is through ``preemptive bailouts.'' Expected future transfers relax the current financial constraints guaranteeing bank solvency and alleviating the limited enforcement friction in the first place. Such transfers are systemic, not being a source of moral hazard.

The economy I consider is a real business cycle model with a banking sector \citep{gertler10,gertler11} and a nonlinear investment technology \citep{lucas71}. The banking system consists of heterogeneous banks that issue deposits to households, invest in the real economy with state-contingent returns, and face survival risk. Financial intermediation is imperfect due to the limited enforcement of deposit contracts and the resulting enforcement constraint faced by individual banks. The enforcement constraint posits that the forward-looking bank value must be at least as great as the value of default---running away with a fraction of assets---at all possible contingencies. The constraint is thus similar to that studied by \cite{kehoe02} in an international real business cycle model with endogenously incomplete markets.

I highlight two distinct sources of the inefficiency of the competitive equilibrium allocation. First, there are pecuniary externalities: individual bankers do not internalize that their private portfolio decisions influence the price of claims on firm profits and the future return on bank assets, affecting both the forward-looking bank value and the value of default of all banks. The equilibrium asset price at $t$ is a function of the aggregate capital stock chosen at $t-1$ and $t$, depending negatively on the former and positively on the latter. Consequently, in the partial derivative sense, greater bank lending to the real sector at $t$ is linked to a greater demand for capital goods and asset price at $t$, increasing the value of default; furthermore, there is a lower asset price and marginal product of capital at $t+1$, decreasing the ex post return on bank assets at $t+1$ and the bank value at $t$. Both effects tighten enforcement constraints of all banks at $t$. On the other hand, by decreasing the asset price at $t+1$, greater bank lending at $t$ decreases the value of default at $t+1$, relaxing future enforcement constraints. Hence, while the former effects are consistent with excessive borrowing and lending in the competitive equilibrium, the latter effect contributes to insufficient borrowing and lending, and in both cases, the effect on bank borrowing is due to the balance sheet identity. There are, moreover, additional externalities, many of which depend on the extent of commitment by the policymaker. A planner that can choose a policy plan at the beginning of time once and for all internalizes the effect of allocations at $t$ on $t-1$ expectations. As a result, the planner with commitment internalizes that greater bank lending at $t$ increases the bank value and relaxes the enforcement constraint at $t-1$, which is due to a positive effect on the asset price and thus on the ex post asset return at $t$. This effect is absent without commitment when the planner must consider how current decisions affect the endogenous state and optimal decisions in the future.

The second type of inefficiencies is the very nature of the limited enforcement friction that restricts bank borrowing and lending compared to the economy without such a friction. Intuitively, if the enforcement constraint is binding at $t$, one can achieve a strict welfare improvement by promising a greater future bank value conditional on survival to $t+1$, relaxing the financial constraint at $t$. Formally, this goal can be achieved by manipulating the future bank value distribution between entrants and survivors. Under the assumption of commitment, the planner internalizes how allocations affect the future distribution, and how the latter impacts the tightness of the current enforcement constraint.

Formally, I characterize the constrained efficient allocation, which results from a planning problem of a benevolent policymaker that maximizes household welfare subject to the competitive equilibrium implementability constraints, but making portfolio decisions on behalf of the banking system. I study this problem both under the assumption of commitment and no commitment on the planner's side. The no-commitment case corresponds to a Markov perfect equilibrium of a noncooperative game between successive policymakers \citep{klein08}. Both types of constrained efficient allocations highlight similar distortions in the competitive equilibrium. However, by construction, the time-consistent planner cannot affect the future bank value distribution except by affecting the future endogenous state. Therefore, the distribution is always taken as given in the time-consistent analysis.

I show how to implement both types of constrained efficient allocations in a regulated competitive equilibrium with two types of policy instruments that address the two types of inefficiencies. The externalities can be corrected either by linear taxes on deposits and loans balanced in the aggregate or by one of these types of taxes rebated lump sum. On the contrary, the minimum state-contingent bank capital requirements are generally insufficient. The problem with the latter is its asymmetric nature that does not allow closing the wedges in the contingencies in which the optimal credit spread is too low. The given bank value distribution can be achieved with entrants/survivors-specific transfers---preemptive bailouts---that either add up to zero or match the aggregate lump-sum transfer that rebates the proceeds from the linear taxes. For most computations, I use global projection methods to fully account for precautionary savings effects when the occasionally binding enforcement constraint is likely to switch from the nonbinding to binding regime.

Quantitatively, in both the Markov perfect and Ramsey equilibria, the enforcement constraint is binding by an order of magnitude less often than in the competitive equilibrium. Both normative arrangements generate sizable consumption-equivalent welfare gains: up to a 0.75\% ergodic distribution average at the commitment allocation. In the Ramsey equilibrium, the optimal bank value distribution is biased towards surviving banks, and there is more borrowing and lending than in the competitive equilibrium. At the same time, the Ramsey allocation has less borrowing and lending than in the unconstrained competitive equilibrium---that is, the competitive equilibrium in the environment without the enforcement friction. Without commitment, internalizing the pecuniary externalities is the key pursuit, as the planner has limited ability to affect the future bank value distribution. With commitment, preemptive bailouts have more power, and the planner engineers a survivors-biased distribution with greater equilibrium borrowing and lending. Transfers to survivors increase around crises.

\paragraph{Related literature}
This paper is related to the literature on financial crises and pecuniary externalities arising from endogenous financial constraints. \cite{lorenzoni08} is one of the first papers to highlight overborrowing in the competitive equilibrium due to a pecuniary externality in the price of capital in a three-period model with two-sided limited commitment and direct finance; \cite{davila18} provide a comprehensive theoretical analysis of pecuniary externalities. \cite{bianchi11} considers a quantitative endowment (small open) economy with two goods and a flow collateral constraint that depends on the relative price of nontradable goods. He shows that overborrowing due to the pecuniary externality can be corrected with a state-contingent debt tax. In the same model, \cite{benigno16} show that policies that relax financial constraints ex post achieve greater welfare than the optimal debt tax since the former can implement the unconstrained first-best allocation. The competitive equilibrium features underborrowing compared to that allocation. Moreover, \cite{schmitt-grohe21} emphasize that there exist reasonable parameterizations under which multiple equilibria arise in that model, including a self-fulfilling equilibrium that features underborrowing. \cite{benigno13} find that underborrowing arises in a related production economy and highlight the importance of ex post policies. In the context of small open and endowment economies with stock collateral constraints, \cite{bianchi18} and \cite{jeanne19} identify overborrowing in the competitive equilibrium compared to the Markov perfect equilibrium and characterize the optimal time-consistent debt tax. The current paper focuses on the implications of pecuniary externalities due to asset prices and asset returns that affect the forward-looking enforcement constraint faced by financial intermediaries.

This paper is also related to the literature on financial crises, bailouts, and optimal financial regulation. \cite{allen04} generalize the model of \cite{diamond83} and show that with complete markets and limited market participation, the competitive equilibrium is either incentive efficient or constrained efficient, defaults and financial crises occur in equilibrium with incomplete contracts, and no regulation is warranted. However, with incomplete markets, there is room for liquidity regulation. \cite{farhi09} show that the competitive equilibrium in the model of \cite{diamond83} is constrained inefficient even with complete markets if agents can engage in hidden trades. In the current paper, there is endogenous market incompleteness that gives rise to pecuniary externalities and inefficient financial intermediation. \cite{farhi12} demonstrate that imperfectly targeted time-consistent accommodating interest-rate policies lead to multiple equilibria, increase correlation in risk-taking behavior by financial intermediaries and sow the seeds of future crises. Regulation in the form of a cap on short-term debt reduces the set of equilibria to a singleton that corresponds to the commitment benchmark. The current paper focuses on the symmetric equilibrium in the banking system to facilitate aggregation and permit tractable positive and normative analysis. Even though preemptive bailouts are imperfectly targeted within entrants and survivors, the symmetric equilibrium is not subject to collective moral hazard. \cite{bianchi16} studies the implications of a pecuniary externality in an equity constraint that depends on the market wage rate and emphasizes the benefits of a systemic debt relief policy---a proportional reduction in debt repayments---that helps relax equity constraints during crises. The objective of relaxing financial constraints is similar to the objective of preemptive bailouts in the current paper, but the latter constitute a somewhat different policy---group-dependent lump-sum transfers provided to banks at $t+1$ to relax financial constraints at $t$. \cite{chari16} develop a model where costly firm bankruptcies occur in the competitive equilibrium, which is both ex ante and ex post efficient if compared to the commitment benchmark. Without commitment, inefficient bailouts will arise, and regulation in the form of a limit on the debt-to-value ratio and the tax on firm size is desirable to achieve a sustainably efficient outcome. In the current paper, the competitive equilibrium is constrained inefficient compared to the commitment benchmark, while preemptive (not actual) bailouts help mitigate the source of endogenous market incompleteness.

As part of smaller quantitative literature, \cite{boissay16} develop a real business cycle model with a banking sector that features an interbank market. High-skilled banks borrow from low-skilled banks and households to lend to firms and may decide to divert borrowed funds to invest in the storage technology subject to diversion costs. A relevant incentive compatibility constraint eliminates the former possibility. The authors briefly discuss the constrained inefficiency of the competitive equilibrium and compute welfare losses. A crucial difference from the current paper is that the bank's problem is static, and the incentive constraint is always binding in equilibrium; therefore, the sources of the inefficiency of the competitive equilibrium are utterly different from those in the current paper. Indeed, in the current paper, the competitive equilibrium is constrained efficient if the enforcement constraint is always binding, and the bank value distribution externality arises because the enforcement constraint is forward looking. \cite{collard17} study locally Ramsey-optimal bank capital requirements and monetary policy. In their model, sufficiently high capital requirements help eliminate risky lending in equilibrium. On the contrary, in the current paper, capital requirements do not generally constitute an effective policy instrument, and their role is different---to force individual banks to internalize pecuniary externalities due to the enforcement constraint. In a continuous-time environment, \cite{ditella19} demonstrates how the possibility of hidden trades in physical capital by intermediaries inflates the asset price and risk exposure of other intermediaries. The constrained efficient allocation can be implemented with a tax on assets, while bank capital requirements are ineffective. \cite{vanderghote21} develops a continuous-time model with nominal rigidities and a banking sector that is similar to that in the current paper but with the capital requirement constraint imposed as part of the environment. The author restricts attention to Markov equilibria and acknowledges the presence of pecuniary externalities, discussing them intuitively and computing the optimal capital requirement numerically. The current paper instead characterizes constrained efficient allocations that do not depend on the presence of specific policy instruments. Indeed, as mentioned above, capital requirements might not be effective for correcting the externalities. Moreover, the current paper identifies the bank value distribution externality and conducts the normative analysis both with and without commitment on the planner's side.

Several papers studied the welfare implications of specific policies in related environments under the assumption that the enforcement constraint is always binding to allow for smooth local approximations \citep{gertler10,gertler11,gertler12,depaoli17}. In the current paper, the enforcement constraint is occasionally binding, and the model is solved using global or quasi-global methods. Moreover, as mentioned above, the competitive equilibrium is constrained efficient if the enforcement constraint is always binding, so in that case, regulation might not be desirable. \cite{gertler20RED} and \cite{akinci22} also use global methods, but they do not study efficiency, restricting the analysis to specific policy rules.

The rest of this paper proceeds as follows. Section \ref{sec: model} describes the theoretical model and defines the competitive equilibrium. Section \ref{sec: normative analysis} conducts the normative analysis. Section \ref{sec: quantitative results} presents quantitative results. Section \ref{sec: conclusion} concludes. \hyperlink{appendices}{Appendices} contain proofs of theoretical results.

\section{Model}\label{sec: model}
Consider a basic version of the economy described in \cite{gertler10} and \cite{gertler11}. Time $t$ is discrete, and the horizon is infinite, so $t\in\mathbb{Z}_+$. There are unit measures of identical households and firms producing final and capital goods. Each household is a family of $f\in(0,1)$ bankers and $1-f$ workers, and there is perfect consumption insurance within a family. Bankers manage banks that intermediate funds between households and final good firms. Crucially, there is limited enforcement of deposit contracts between households and banks, which gives rise to an endogenous financial friction. Without that friction, the economy collapses to a standard real business cycle model. At each date $t>0$, there is uncertainty regarding the state of nature $s_{t}\in{}S$, and $s_{0}\in{}S$ is fixed. For simplicity, we can think of $S$ being finite. A history of states is $s^{t}=(s_{0},s_{1},\dots,s_{t})\in{}S^{t}$, where $S^{t}\equiv{}S^{t-1}\times{}S$ with $S^{0}\equiv{}S$, and the probability of a history $s^{t}$ is $\pi_{t}(s^{t})$. We will keep the history dependence implicit when possible.

\subsection{Households}
On behalf of a family, the head of the household decides how much to consume $C_{t}$, save in one-period deposits $\frac{D_{t+1}}{R_{t}}$ with the gross return $R_{t}$, and how much labor $L_{t}$ to supply at the wage rate $W_{t}$. The budget constraint of the household is
\begin{equation*}
    C_{t}+\frac{D_{t+1}}{R_{t}}\le{W_{t}}L_{t}+D_{t}+\Pi_{t},
\end{equation*}
where $\Pi_{t}$ denotes net transfers from the ownership of banks and firms.

The household's preferences are represented by $\mathbb{E}_{0}[\sum_{t=0}^{\infty}\beta^{t}U(C_{t},L_{t})]$, where $\beta\in(0,1)$, $U:\mathbb{R}^2_{+}\to\mathbb{R}$ is twice continuously differentiable and strictly concave with $U_C>0$, $U_L<0$, and $\lim_{C\to{0}}U_C(C,L)=\infty$ for all $L\ge{0}$. The necessary conditions for optimality include the budget constraint holding as equality, the labor supply condition that links the wage to the marginal rate of substitution of consumption for leisure \eqref{eq: household labor supply}, and the Euler equation that prices deposits \eqref{eq: household Euler deposits}.
\begin{align}
    W_{t}&=-\frac{U_{L,t}}{U_{C,t}},\label{eq: household labor supply}\\
    U_{C,t}&=\beta{}R_{t}\mathbb{E}_{t}(U_{C,t+1}).\label{eq: household Euler deposits}
\end{align}
Combined with initial and transversality conditions on $\{D_{t}\}$, the above equations are also sufficient to determine the household's optimal plan given prices and government policies.

\subsection{Bankers}
A banker manages a bank that invests net worth $n_{t}$ and deposits $\frac{d_{t+1}}{R_{t}}$ into firms' equity $k_{t+1}$ at the price $Q_{t}$. The bank's balance sheet constraint is then
\begin{equation*}
    Q_{t}k_{t+1}=n_{t}+\frac{d_{t+1}}{R_{t}}.
\end{equation*}

Bankers are assumed to stay in the banking business for a finite expected period of time. Specifically, a banker in period $t$ remains to be a banker in period $t+1$ with the probability $\sigma\in[0,1)$ and becomes a worker with the probability $1-\sigma$. Hence, the expected lifetime of the banking business is $\frac{1}{1-\sigma}$. A banker that exits transfers the accumulated net worth to their household. Accordingly, $(1-\sigma)f$ workers start a banking business each period, being endowed with a startup net worth $n^0_{t}>0$ by their households. The future net worth is the difference between the ex post returns on assets and liabilities:
\begin{equation*}
    n_{t+1}=X_{t+1}k_{t+1}-d_{t+1},
\end{equation*}
where $X_{t}$ is the gross payoff per unit of capital stock.

Let $v_{t}$ denote the bank value. The bank value satisfies a stochastic difference equation:
\begin{equation*}
    v_{t}=\mathbb{E}_{t}\{\Lambda_{t,t+1}[(1-\sigma)n_{t+1}+\sigma{}v_{t+1}]\},
\end{equation*}
where $\Lambda_{t,t+1}\equiv\beta\frac{U_{C,t+1}}{U_{C,t}}$ is the representative household's stochastic discount factor.

We assume that at the end of a period $t$, a banker could divert a fraction $\theta\in[0,1]$ of assets to their household. In that case, the bank would default, while other households could recover only the remaining fraction $1-\theta$ of assets. Consequently, other households will be willing to lend to the bank only if the following incentive compatibility constraint holds:
\begin{equation*}
    v_{t}\ge\theta{}Q_{t}k_{t+1}.
\end{equation*}
This constraint is essentially an enforcement constraint (EC) of the type studied in \cite{kehoe02}. Note that $v_{t}$ depends on infinitely many future controls, similar to the class of problems studied by \citet{marcet19}. Crucially, however, the constraint set is linear in $n_{t}$, which simplifies the matters significantly.

The banker's problem is
\begin{equation*}
    \max_{\{d_{t+1},k_{t+1},v_{t}\}}v_{0}
\end{equation*}
subject to the nonnegativity, balance sheet, net worth, bank value, and enforcement constraints. Let $\Tilde{\nu}_{t}(s^{t})$, $\gamma_{t}(s^{t})$, and $\Tilde{\lambda}_{t}(s^{t})$ denote the Lagrange multipliers on the balance sheet, bank value, and enforcement constraints for a history $s^{t}$ normalized by $(\beta\sigma)^{t}\pi_{t}(s^{t})\frac{U_{C,t}(s^{t})}{U_{C,0}}$. Define also the scaled multipliers $\nu_{t}\equiv\frac{\Tilde{\nu}_{t}}{\gamma_{t-1}}$ and $\lambda_{t}\equiv\frac{\Tilde{\lambda}_{t}}{\gamma_{t-1}}$. The next proposition characterizes the solution to the banker's problem.
\begin{proposition}\label{prop: banker's problem}
    The solution to the banker's problem is characterized by the Euler equations
    \begin{align}
        \nu_{t}&=(1+\lambda_{t})\mathbb{E}_{t}[\Lambda_{t,t+1}(1-\sigma+\sigma\nu_{t+1})]R_{t},\label{eq: bank Euler deposits DE}\\
        \theta\lambda_{t}+\nu_{t}&=(1+\lambda_{t})\mathbb{E}_{t}\left[\Lambda_{t,t+1}(1-\sigma+\sigma\nu_{t+1})\frac{X_{t+1}}{Q_{t}}\right],\label{eq: bank Euler capital DE}
    \end{align}
    and the complementary slackness conditions
    \begin{equation*}
        0=\lambda_{t}(v_{t}-\theta{}Q_{t}k_{t+1}),\qquad
        \lambda_{t}\ge{}0,\qquad
        v_{t}\ge\theta{}Q_{t}k_{t+1},
    \end{equation*}
    where the value function is $v_{t}=\nu_{t}n_{t}$. The scaled multipliers $\lambda_{t}$ and $\nu_{t}$ and the ratios $\frac{d_{t+1}}{n_{t}}$ and $\frac{k_{t+1}}{n_{t}}$ are independent of $n_{t}$.
\end{proposition}
\begin{proof}
    See Appendix \ref{sec: proof banker's problem}.
\end{proof}

According to the expression for the value function, the marginal value of net worth equals the scaled multiplier $\nu_{t}$. The intuition is that net worth is more valuable when the balance sheet constraint is tighter: the greater the original multiplier on the balance sheet constraint $\Tilde{\nu}_{t}$, the greater the marginal value of net worth $\nu_{t}$.

As shown in Appendix \ref{sec: proof banker's problem}, at the optimal plan, $\gamma_{t}=1+\sum_{j=0}^{t}\lambda_{j}$. Remember that $\gamma_{t}$ affects the scaled multipliers $\nu_{t}$ and $\lambda_{t}$; therefore, similar to \cite{kehoe02} and \cite{marcet19}, the solution to the banker's problem depends on the history of Lagrange multipliers associated with the EC. At the same time, the scaled multipliers $\nu_{t}$ and $\lambda_{t}$ are sufficient statistics for the characterization of the optimal plan.

The KKT conditions \eqref{eq: bank Euler deposits DE} and \eqref{eq: bank Euler capital DE} imply that the risk-adjusted credit spread $\mathbb{E}_{t}[\Lambda_{t,t+1}(1-\sigma+\sigma\nu_{t+1})(\frac{X_{t+1}}{Q_{t}}-R_{t})]$ is entirely determined by the scaled multiplier $\lambda_{t}$. The tighter the EC, the greater the $\lambda_{t}$, and the greater the spread.

By Proposition \ref{prop: banker's problem},
\begin{equation}\label{eq: banking sector value DE}
    V_{t}=\nu_{t}N_{t}.
\end{equation}
The signs of all individual ECs are equivalent to the sign of the aggregate EC. Hence, it is enough to consider the aggregate complementary slackness conditions
\begin{equation}\label{eq: banking sector complementary slackness DE}
    0=\lambda_{t}(V_{t}-\theta{}Q_{t}K_{t+1}),\qquad
    \lambda_{t}\ge{0},\qquad
    V_{t}\ge\theta{}Q_{t}K_{t+1}.
\end{equation}
The banking sector balance sheet is
\begin{equation}\label{eq: banking sector balance sheet}
    Q_{t}K_{t+1}=N_{t}+\frac{D_{t+1}}{R_{t}}.
\end{equation}
The aggregate net worth $N_{t}$ is a sum of the aggregate net worth of survivors $N^1_{t}$ and entrants $N^0_{t}$. A fraction $\sigma$ of old banks survive, so their aggregate net worth is $N^1_{t}=\sigma(X_{t}K_{t}-D_{t})$. The aggregate endowment of entrants is $N^0_{t}=\Bar{N}+\omega{}Q_{t}K_{t}$, where $(\Bar{N},\omega)\in\mathbb{R}^2_+$. Hence,
\begin{equation}\label{eq: banking sector net worth}
    N_{t}=\sigma(X_{t}K_{t}-D_{t})+\Bar{N}+\omega{}Q_{t}K_{t}.
\end{equation}
A necessary condition for the existence of a deterministic steady state is $\sigma{R}<1$, which ensures that the initial net worth $N_{0}>0$---determined by the initial conditions---vanishes as $t\to\infty$. Taking into account \eqref{eq: household Euler deposits}, the former condition is equivalent to $\sigma<\beta$. Quantitatively, $\sigma<\beta$ is also necessary for the existence of an ergodic distribution: $\sigma{R_{t}}<1$ must hold ``on average'' to have $\lim_{t\to\infty}\mathbb{E}(N_{t})\in\mathbb{R}_{++}$.

\subsection{Firms}
The economy is populated by firms that produce final and capital goods.

\subsubsection{Final good producers}
Firms that produce the final good demand labor $L_{t}$ and purchase machines and equipment $K_{t}$ from capital good producers. The technology is represented by a production function $F:\mathbb{R}^2_+\to\mathbb{R}_+$, which is twice continuously differentiable, satisfies Inada conditions, and exhibits constant returns to scale. Firms rely on external financing from banks to purchase capital goods by offering state-contingent securities, which correspond to the quantity of capital goods demanded. By no-arbitrage, the price of both securities and capital goods equals $Q_{t}$. The objective of the firm is then
\begin{equation*}
    \max_{K_{t},L_{t}}A_{t}F(\xi_{t}K_{t},L_{t})+Q_{t}(1-\delta)\xi_{t}K_{t}-X_{t}K_{t}-W_{t}L_{t},
\end{equation*}
where $A_{t}$ is the total factor productivity (TFP), $\xi_{t}$ represents capital quality, and $\delta\in[0,1]$ is the depreciation rate. Both $\{A_{t}\}$ and $\{\xi_{t}\}$ are exogenous stochastic processes. Profit maximization implies
\begin{align}
    X_{t}&=[A_{t}F_{K,t}+Q_{t}(1-\delta)]\xi_{t},\label{eq: firm capital demand}\\
    W_{t}&=A_{t}F_{L,t}.\label{eq: firm labor demand}
\end{align}

\subsubsection{Capital good producers}
Capital goods are produced according to a production technology $(I,K)\to{}f(I,K)$, where $f:\mathbb{R}^2_+\to\mathbb{R}$ is strictly increasing in $I$, increasing in $K$, and concave. An example of such a technology is described in \cite{lucas71}. The firm's problem is static:
\begin{equation*}
    \max_{I_{t}}Q_{t}f(I_{t},K_{t})-I_{t},
\end{equation*}
which implies
\begin{equation}
    Q_{t}=\frac{1}{f_{I,t}}.\label{eq: firm capital supply}
\end{equation}

\subsection{Market clearing}
The capital, securities, and final good markets clear as follows:
\begin{align}
    K_{t+1}&=(1-\delta)\xi_{t}K_{t}+f(I_{t},K_{t}),\label{eq: market clearing capital}\\
    A_{t}F(\xi_{t}K_{t},L_{t})&=C_{t}+I_{t}.\label{eq: market clearing final good}
\end{align}

\subsection{Competitive equilibrium}\label{sec: competitive equilibrium}
We now define the competitive equilibrium (CE) in the unregulated economy.
\begin{definition}\label{def: CE}
    Given initial conditions $D_{0}$, $K_{0}$, transversality conditions, and exogenous stochastic processes $\{A_{t},\xi_{t}\}$, the CE is represented by the following measurable functions that map $S^{t}$ to $\mathbb{R}$ for all $t\ge{0}$:
    \begin{itemize}
        \item allocations $C_{t}$, $D_{t+1}$, $I_{t}$, $K_{t+1}$, $L_{t}$, $N_{t}$;
        \item prices $Q_{t}$, $R_{t}$, $W_{t}$, $X_{t}$;
        \item transformed Lagrange multipliers $\lambda_{t}$, $\nu_{t}$.
    \end{itemize}
    The functions are consistent with \eqref{eq: household labor supply}--\eqref{eq: market clearing final good} for all $t\ge{0}$.
\end{definition}
The linearity of the banker's problem allows the construction of a set of allocations of individual bankers consistent with the CE.

\section{Normative analysis}\label{sec: normative analysis}
This section studies the problem of a benevolent social planner who will maximize household welfare, internalizing the determination of market prices and making the optimal portfolio decisions on behalf of the banking system subject to the aggregate EC. We will characterize the constrained efficient allocation under commitment (CEA) and the Markov-perfect constrained efficient allocation (MCEA). We will show how to implement the CEA and MCEA in the regulated CE with either affine taxes on bank assets and liabilities or state-contingent capital requirements to address the pecuniary externalities and bank entrants/survivors-specific transfers to achieve the targeted bank value distribution.

\subsection{Sources of inefficiency}
To proceed with the formal characterization of the CEA and MCEA, we must derive the aggregate EC of the banking system. Doing so will also clarify the nature of distortions in the CE on an intuitive level.

Let us index the existing bankers with $i\in[0,f]$. We can assume without loss of generality that survivors are always in the $[0,\sigma{}f]$ interval, and entrants are in the $(\sigma{}f,f]$ interval. Hence, the indices of $(1-\sigma)\sigma{}f$ current survivors that will exit next period will be filled by $\sigma(1-\sigma)f$ current entrants that will survive to the next period. Let $v^1_{i,t+1}(s^{t+1})$ denote the bank value of the banker $i$ conditional on survival from $s^{t}$ to $s^{t+1}$. Let $\Delta_{t}\equiv\frac{V^1_{t}}{V_{t}}$, where $V^1_{t}\equiv\int_0^{\sigma{}f}v_{i,t}\dd{i}$ is the aggregate bank value of survivors. It follows that the aggregate bank value of the banking system satisfies
\begin{align}
    V_{t}&\equiv\int_0^{f}v_{i,t}\dd{i}\nonumber\\
    &=\mathbb{E}_{t}\left\{\Lambda_{t,t+1}\left[(1-\sigma)\int_0^{f}n_{i,t+1}\dd{i}+\sigma{}\int_0^{f}v^1_{i,t+1}\dd{i}\right]\right\}\nonumber\\
    &=\mathbb{E}_{t}\left\{\Lambda_{t,t+1}\left[(1-\sigma)(X_{t+1}K_{t+1}-D_{t+1})+\int_0^{\sigma{}f}v_{i,t+1}\dd{i}\right]\right\}\nonumber\\
    &=\mathbb{E}_{t}\{\Lambda_{t,t+1}[(1-\sigma)(X_{t+1}K_{t+1}-D_{t+1})+\Delta_{t+1}V_{t+1}]\},\label{eq: aggregate bank value}
\end{align}
where the second equality follows from the definition of the individual bank value of the banker $i\in[0,f]$, the third equality follows from the fact that each banker $i$ has the same probability of survival from $s^{t}$ to $s^{t+1}$, and the fourth equality follows from definitions of $V^1$ and $\Delta$.

The aggregate EC is
\begin{equation}
    V_{t}\ge\theta{}Q_{t}K_{t+1}.\label{eq: aggregate EC}
\end{equation}
Further using \eqref{eq: firm capital demand} and substituting \eqref{eq: aggregate bank value} in \eqref{eq: aggregate EC}, we obtain
\begin{multline*}
    \mathbb{E}_{t}\Bigl\{\Lambda_{t,t+1}\Bigl[(1-\sigma)\{[A_{t+1}F_K(\xi_{t+1}K_{t+1},L_{t+1})+Q(K_{t+1},K_{t+2},\xi_{t+1})(1-\delta)]\xi_{t+1}K_{t+1}-D_{t+1}\}\\
    +\Delta_{t+1}V_{t+1}\Bigr]\Bigr\}\ge\theta{}Q(K_{t},K_{t+1},\xi_{t})K_{t+1},
\end{multline*}
where the asset price function $Q$ is defined by \eqref{eq: firm capital supply} and \eqref{eq: market clearing capital} as
\begin{equation*}
    Q(\underset{-}{K_{t}},\underset{+}{K_{t+1}},\xi_{t})=\left[\Phi'\left(\Phi^{-1}\left(\frac{K_{t+1}}{K_{t}}-(1-\delta)\xi_{t}\right)\right)\right]^{-1}.
\end{equation*}
The function $Q$ is decreasing in the first argument and increasing in the second argument, which follows from $\Phi$ being strictly increasing and concave.

There are two broad sources of potential distortions in the CE allocation. The first, highlighted in red, arises because individual bankers do not internalize how their asset allocations affect the current asset price and the future asset returns. The second, highlighted in green, reflects that the future continuation value of the banking system conditional on survival might be inefficiently low.

The first type of distortions reflects pecuniary externalities working through the asset price $Q$ and the asset payoff $X$, affecting both the bank value and the value of default---running away with a fraction of assets. First, private bankers do not internalize that higher investment in the real sector---higher $K_{t+1}$ in the aggregate---decreases the future asset returns by decreasing both the future marginal product of capital and the future asset price, which, in turn, decreases the current bank value and makes the ECs of all banks more likely to be binding at $t$. Second, individual bankers do not internalize that greater $K_{t+1}$ increases the current asset price $Q_{t}$, making the default option more attractive and further increasing the probability that ECs of all bankers are binding at $t$. Third, since greater $K_{t+1}$ decreases the future asset price, it has a negative effect on the future value of default, relaxing the future ECs. Fourth, from the perspective of the planner that has commitment, a higher $K_{t+1}$ increases the $t-1$ expectation of the current asset return, thus relaxing the EC at $t-1$. Fifth, from the perspective of the planner that does not have commitment and limits its policies to Markovian ones, the changes in $D_{t+1}$ and $K_{t+1}$ are the changes in the endogenous state variables of the ``future'' planner, having multiple additional effects through the future policy functions. Therefore, the private portfolio decisions might be distorted through multiple channels working in opposite directions, some of which depend on the assumption of commitment from the planner's side. We will study these channels in more detail in the following subsections.

The nature of the second type of potential inefficiencies is linked to how the future bank value conditional on survival affects the current value of the banking system. From the perspective of an individual banker, the continuation value is a product of the constant survival probability $\sigma$ and the future bank value $v_{t+1}$. From the planner's perspective, the aggregate continuation value equals the aggregate bank value of the survived banks $V^1_{t+1}$, which is a state-contingent share $\Delta_{t+1}$ of the aggregate future bank value $V_{t+1}$. If the planner could choose $\{\Delta_{t}\}$, it would generally be optimal to increase it in all contingencies to relax the aggregate EC and thus expand the feasible set, potentially leading to welfare gains.

We are now ready to proceed with the formal characterization of the constrained efficient allocations, both with and without commitment.

\subsection{Constrained efficient allocation under commitment}
Consider the sequential planning problem of optimizing the household welfare by choosing infinite sequences of history-contingent allocations at $t=0$ subject to relevant infinite sequences of history-contingent CE implementability constraints. By Definition \ref{def: CE}, the complete set of CE implementability conditions is \eqref{eq: household labor supply}--\eqref{eq: market clearing final good}. Since we let the planner optimize on behalf of the banking system, the constraints \eqref{eq: bank Euler deposits DE}, \eqref{eq: bank Euler capital DE}, and \eqref{eq: banking sector complementary slackness DE} are not applicable. Consequently, we replace (bankers aggregate EC) with the definition of the aggregate bank value \eqref{eq: aggregate bank value} and the aggregate EC \eqref{eq: aggregate EC}. We can use \eqref{eq: household Euler deposits}, \eqref{eq: banking sector net worth}, \eqref{eq: firm labor demand}, \eqref{eq: firm capital demand}, \eqref{eq: firm capital supply}, and \eqref{eq: market clearing capital} to solve for $R_{t}$, $N_{t}$, $W_{t}$, $X_{t}$, $Q_{t}$, and $I_{t}$, respectively. It is also convenient to define the investment and asset price functions $I$ and $Q$. (The latter has already been defined in the previous subsection.) The investment function $I$ is defined based on \eqref{eq: market clearing capital} as
\begin{equation*}
    I(\underset{-^*}{K_{t}},\underset{+}{K_{t+1}},\xi_{t})=\Phi^{-1}\left(\frac{K_{t+1}}{K_{t}}-(1-\delta)\xi_{t}\right)K_{t},
\end{equation*}
where $^*$ in $-^*$ indicates a numerical statement, although it is true under any reasonable calibration.

Before describing the planning problem, we must decide how to handle $\Delta_{t+1}$ appearing in \eqref{eq: aggregate bank value}. By definition, we must have $\Delta_{t}(s^{t})\in[0,1)$ for all $t\ge{}0$ and $s^{t}\in{}S^{t}$. To see that the right bound is not included, note that otherwise we would have $v_{i,t}=0$ for all entering bankers $i\in(\sigma{}f,f]$. By the individual EC, we would then have $Q_{t}k_{i,t+1}=0$ for all such $i$, implying that all entrants could not operate. Suppose the planner considers $\{\Delta_{t}\}$ as a control variable. Since the latter affects the continuation value in the EC only, it may be optimal to set $\Delta_{t+1}(s^{t+1})\to{}1$ if the EC is binding at $s^{t}$. In such a case, the maximum cannot be attained. To avoid this problem, let us, first, define the CE-consistent bank value distribution $\{\widehat{\sigma}^1_{t}\}$, where
\begin{equation*}
    \widehat{\sigma}^1_{t}\equiv\frac{\sigma(X_{t}K_{t}-D_{t})}{N_{t}}.
\end{equation*}
We will then conduct the analysis under the assumption that $\{\Delta_{t}\}$ is either given or satisfies $\{\Delta_{t}\}=\{\widehat{\sigma}^1_{t}\}$. Since the feasible set for $\{\Delta_{t}\}$ is a space of sequences of functions that map to an open unit interval, we can explore the implications of alternative distributions $\{\Delta_{t}\}$ quantitatively in a straightforward manner.

The sequential planning problem is, therefore,
\begin{equation*}
    \max_{\{C_{t},D_{t+1},K_{t+1},L_{t},V_{t}\}}\mathbb{E}_{0}\left[\sum_{t=0}^{\infty}\beta^{t}U(C_{t},L_{t})\right]
\end{equation*}
subject to
\begin{align*}
    0&=N_{t}-Q(K_{t},K_{t+1},\xi_{t})K_{t+1}+\beta\mathbb{E}_{t}\left(\frac{U_C(C_{t+1},L_{t+1})}{U_C(C_{t},L_{t})}\right)D_{t+1},\\
    0&=\mathbb{E}_{t}\left\{\beta\frac{U_C(C_{t+1},L_{t+1})}{U_C(C_{t},L_{t})}[(1-\sigma)(X_{t+1}K_{t+1}-D_{t+1})+\Delta_{t+1}V_{t+1}]\right\}-V_{t},\\
    0&\le{}V_{t}-\theta{}Q(K_{t},K_{t+1},\xi_{t})K_{t+1},\\
    0&=U_C(C_{t},L_{t})A_{t}F_{L}(\xi_{t}K_{t},L_{t})+U_L(C_{t},L_{t}),\\
    0&=A_{t}F(\xi_{t}K_{t},L_{t})-C_{t}-I(K_{t},K_{t+1},\xi_{t}),
\end{align*}
where
\begin{gather*}
    X_{t}\equiv[A_{t}F_{K}(\xi_{t}K_{t},L_{t})+Q(K_{t},K_{t+1},\xi_{t})(1-\delta)]\xi_{t},\\
    N_{t}\equiv\sigma(X_{t}K_{t}-D_{t})+\Bar{N}+\omega{}Q(K_{t},K_{t+1},\xi_{t})K_{t},
\end{gather*}
and $\{\Delta_{t}\}$ is either given or satisfies
\begin{equation*}
    \Delta_{t}=\Delta(\underset{-}{D_{t}},\underset{-^*}{K_{t}},\underset{+^*}{K_{t+1}},\underset{+}{L_{t}},A_{t},\xi_{t})=\widehat{\sigma}^1_{t},
\end{equation*}
where we again use notation $-^*$ and $+^*$ to indicate numerical statements.

Let us denote the Lagrange multipliers on the planner's constraints---normalized by $\beta^{t}\pi(s^{t})$---as $\widetilde{\nu}_{t}$, $\widetilde{\gamma}_{t}$, $\widetilde{\lambda}_{t}$, $\lambda^L_{t}$, and $\lambda^Y_{t}$, respectively. Define $\nu_{t}\equiv\frac{\widetilde{\nu}_{t}}{U_{C,t}}$, $\lambda_{t}\equiv\frac{\widetilde{\lambda}_{t}}{U_{C,t}}$, and $\gamma_{t}\equiv\frac{\widetilde{\gamma}_{t}}{U_{C,t}}$. As in the CE, define $\widehat{x}_{t}\equiv\frac{x_{t}}{\gamma_{t}}$ and $\Bar{x}_{t}\equiv\frac{\widehat{x}_{t}}{1-\widehat{\lambda}_{t}}$ for $x\in\{\nu,\lambda,\lambda^L,\lambda^Y\}$.

As discussed in the previous subsection, there are multiple potential sources of inefficiency of the CE allocation. The next proposition provides a formal validation.
\begin{proposition}\label{prop: CEA characterization}
    The CE allocation is generally inefficient compared to the CEA. The CEA analogs of the Euler equations \eqref{eq: bank Euler deposits DE} and \eqref{eq: bank Euler capital DE} are
    \begin{align*}
        \Tilde{\nu}_{t}&=(1+\Tilde{\lambda}_{t})\mathbb{E}_{t}\left[\Lambda_{t,t+1}\left(1-\sigma+\Delta_{t+1}\sigma\Tilde{\nu}_{t+1}\underbrace{-\frac{\partial\Delta_{t+1}}{\partial{}D_{t+1}}V_{t+1}}_{\text{future distribution}\;(+)}\right)\right]R_{t},\\
        \theta\Tilde{\lambda}_{t}+\Tilde{\nu}_{t}&=(1+\Tilde{\lambda}_{t})\mathbb{E}_{t}[\Lambda_{t,t+1}(1-\sigma+\Delta_{t+1}\sigma\Tilde{\nu}_{t+1})\frac{X_{t+1}}{Q_{t}}]+\Psi^K_{t},
    \end{align*}
    where $\Psi^K_{t}$ satisfies
    \begin{align*}
        Q_{t}\Psi^K_{t}&\equiv\underbrace{\Tilde{\nu}_{t}Q_{2,t}\{[\sigma(1-\delta)\xi_{t}+\omega]K_{t}-K_{t+1}\}}_{\text{balance sheet}\;(-^*)}+\underbrace{(1+\Tilde{\lambda}_{t})\mathbb{E}_{t}\left(\Lambda_{t,t+1}\frac{\partial\Delta_{t+1}}{\partial{}K_{t+1}}V_{t+1}\right)}_{\text{future distribution}\;(-^*)}\\
        &\quad\underbrace{-\Tilde{\lambda}_{t}\theta{}Q_{2,t}K_{t+1}}_{\text{value of default}\;(-)}\underbrace{-\frac{\Tilde{\lambda}^Y_{t}}{U_{C,t}}I_{2,t}}_{\text{consumption}\;(-^*)}+\underbrace{\frac{\bm{1}_\mathbb{N}(t)}{\Delta_{t}}\Biggl[\underbrace{(1-\sigma)Q_{2,t}(1-\delta)\xi_{t}K_{t}}_{\text{asset return}\;(+)}+\underbrace{\frac{\partial\Delta_{t}}{\partial{}K_{t+1}}V_{t}}_{\text{distribution}\;(+^*)}\Biggr]}_\text{$t-1$ expectations}\\
        &\quad+\underbrace{(1+\Tilde{\lambda}_{t})\mathbb{E}_{t}\{\Lambda_{t,t+1}(1-\sigma+\Delta_{t+1}\sigma\Tilde{\nu}_{t+1})[A_{t+1}F_{KK,t+1}\xi_{t+1}+Q_{1,t+1}(1-\delta)]\xi_{t+1}K_{t+1}\}}_{\text{future asset return}\;(-)}\\
        &\quad+(1+\Tilde{\lambda}_{t})\mathbb{E}_{t}\Biggl(\Lambda_{t,t+1}\Delta_{t+1}\Biggl\{\underbrace{\Tilde{\nu}_{t+1}[\omega(Q_{1,t+1}K_{t+1}+Q_{t+1})-Q_{1,t+1}K_{t+2}]}_{\text{future balance sheet}\;(+^*)}\\
        &\quad\underbrace{-\Tilde{\lambda}_{t+1}\theta{}Q_{1,t+1}K_{t+2}}_{\text{future value of default}\;(+)}+\underbrace{\Tilde{\lambda}^L_{t+1}A_{t+1}F_{KL,t+1}\xi_{t+1}}_{\text{future wage}\;(+^*)}+\underbrace{\frac{\Tilde{\lambda}^Y_{t+1}}{U_{C,t+1}}(A_{t+1}F_{K,t+1}\xi_{t+1}-I_{1,t+1})}_{\text{future consumption}\;(+^*)}\Biggr\}\Biggr).
    \end{align*}
    Moreover, the following holds.
    \begin{enumerate}
        \item If \eqref{eq: aggregate EC} at the CEA is either binding almost surely (a.s.) or nonbinding a.s. for all $t\ge{}0$, then the CEA is time consistent. Otherwise, it is generally time inconsistent.
        \item If $\{\Delta_{t}\}=\{\widehat{\sigma}^1_{t}\}$ and \eqref{eq: aggregate EC} at the CEA is binding a.s. for all $t\ge{}0$, and the CEA satisfies $\mathbb{E}_{t}[\Lambda_{t,t+1}f_{t+1}(\frac{X_{t+1}}{Q_{t}}-R_{t})]\ge{}0$ for all $\{f_{t}\}_{t\ge{}1}$ with $f_{t}:S^{t}\to\mathbb{R}_{++}$, then the CEA equals the CE allocation, that is, the latter is constrained efficient.
        \item Given $\{\Delta_{t}\}$, for all $t\ge{}0$, let $\Bar{S}^{t}\subseteq{}S^{t}$ be the set of histories at which \eqref{eq: aggregate EC} is strictly binding---in the sense that the corresponding Lagrange multiplier is positive. If $\Bar{S}^{t}$ is of positive measure at least for some $t\ge{}0$, then there exists $\{\widetilde{\sigma}^1_{t}\}$ with $\widetilde{\sigma}^1_{t}(s^{t})\in[\Delta_{t}(s^{t}),1)$ for all $(t,s^{t})$ such that $\{\widetilde{\sigma}^1_{t}\}$ is strictly preferred to $\{\Delta_{t}\}$.
    \end{enumerate}
\end{proposition}
\begin{proof}
    See Appendix \ref{sec: proof CEA characterization}.
\end{proof}

Our first observation is that if the planner takes the distribution as given, there is no distortion in the choice of deposits, consistent with our intuitive analysis in the previous subsection.\footnote{One might notice a slight difference in the Euler equations: instead of $\sigma\Tilde{\nu}_{t+1}$ in Proposition \ref{prop: banker's problem} we have $\Delta_{t+1}\sigma\Tilde{\nu}_{t+1}$ in Proposition \ref{prop: CEA characterization}---this difference is solely due to how the Lagrange multiplier on the bank value constraint is related to the multiplier on the EC. All original multipliers are stationary in the CEA, unlike in the CE. If one writes the deposit Euler equation in terms of the original multipliers, it will be symbolically equivalent to that in the CE, as one can verify in the proofs of Propositions \ref{prop: banker's problem} and \ref{prop: CEA characterization}.} If the planner internalizes the determination of the distribution, a wedge between the deposit Euler equations does appear: the social marginal cost of deposits is greater than the private marginal cost because the planner understands that greater borrowing at $t$ has a negative effect on the future net worth of survived banks and, therefore, on their relative bank value. At the same time, the presence of the wedge in the Euler equation should not necessarily lead to overborrowing in the CE because the deposit Euler equation is essentially a fixed-point equation in the transformed multiplier $\{\Tilde{\nu}_{t}\}$ conditional on other variables, and the multipliers in the CE and CEA are generally different.

The wedge between the asset Euler equations $\Psi^K_{t}$ consists of multiple terms with opposing effects on the sign of the wedge. If $\{\Delta_{t}\}=\{\widehat{\sigma}^1_{t}\}$, there are two terms (highlighted in green) capturing the effect of the choice of capital on the bank value distribution. On the one hand, greater capital negatively affects the future distribution through the negative effect on the future asset price and asset return. On the other hand, greater capital positively affects the $t-1$ expectation of the distribution at $t$ through the positive effect on the current asset price. Both effects rely on the nature of commitment.

Consider the remaining terms that do not depend on the ability to affect the distribution. Greater capital affects the current asset price positively, increasing the ex post asset return and net worth of both survivors and entrants (the liability side) while also directly increasing the value of bank assets. These two balance sheet margins typically have a negative net effect on the planner's marginal value of capital. By increasing bank assets, greater capital immediately increases the value of default, thus negatively affecting the marginal benefit of capital. Moreover, it increases investment and lowers consumption, generating an additional negative partial effect. The final negative effect is due to the negative impact of greater capital on the future asset price and asset return.

There are several positive effects. With commitment, greater capital and a greater asset price at $t$ affect the $t-1$ expectation of the current asset return positively, increasing the bank value and relaxing the EC at $t-1$. Furthermore, the balance sheet, value of default, and consumption channels described in the previous paragraph have their future counterparts since the asset price function depends on both the beginning-of-the-period and end-of-the-period capital stock. Greater capital has a negative effect on the future asset price; therefore, the future balance sheet, value of default, and consumption effects have positive signs. Furthermore, greater capital increases the future marginal product of labor and the wage rate, having an additional positive effect.

The inefficiency of the CE allocation relative to the CEA and the fact that the planner chooses allocations that must be consistent with the forward-looking household Euler equation \eqref{eq: household Euler deposits} and the definition of the forward-looking aggregate bank value \eqref{eq: aggregate bank value} imply that the CEA is generally time inconsistent. There is a special case when the CEA is time consistent: it happens if the EC is either always binding or always nonbinding at the CEA. In such a case, the implementability constraints completely determine the CEA. These constraints can be formulated recursively as a system of functional equations on the state space $(D,K,z)$; therefore, the CEA must be time consistent. If this situation occurs with $\{\Delta_{t}\}=\{\widehat{\sigma}^1_{t}\}$, the CEA implementability constraints are necessary for the CE. They are, moreover, sufficient if the expected credit spread discounted with the pricing kernel $\Lambda_{t,t+1}f_{t+1}$ for positive-valued $f$ is nonnegative in the CEA. (This condition guarantees that the CE Lagrange multiplier $\Tilde{\lambda}_{t}$ is nonnegative.) The described situation does not arise quantitatively: the EC is only occasionally binding. Nevertheless, this result has an implication for computing macro-banking models similar to that in this paper. If we computed such a model ignoring the occasionally binding constraint, assuming that it is always binding, we would not be able to identify the externalities and would wrongly conclude that the CE allocation is efficient. Note that the issue here is not in the order of local approximation---the allocations would seem identical independently of the order---but in accounting for the occasionally binding constraint properly.

The final part of Proposition \ref{prop: CEA characterization} states that for a given distribution $\{\Delta_{t}\}$, we can generally find an alternative distribution $\{\widetilde{\sigma}^1_{t}\}$ which is at least weakly preferred to $\{\Delta_{t}\}$ as long as the CEA at the original distribution has contingencies in which the EC is binding. The alternative distribution increases the future bank value of survivors, which automatically increases the current bank value both at $s^{t}$ and the preceding contingencies, relaxing the EC at those contingencies and expanding the planner's feasible set. Again, this argument relies on the nature of commitment: the planner relaxes the EC at $t$ by promising a more survivors-biased distribution at $t+1$, bearing similarity with forward guidance for monetary policy. Note also that ex post, the planner is indifferent between honoring such promises or not because $\{\Delta_{t}\}$ affects the planner's constraints only through the continuation value in the forward-looking bank value. In other words, affecting the bank value distribution by itself is \emph{not} a source of time inconsistency.

\subsection{Markov perfect equilibrium}
Since the CEA is generally time inconsistent, a thorough and complete investigation of the constrained efficiency of the economy considered in this paper requires exploring the implications of the lack of commitment by the policymaker. To do so, we will study a Markov perfect equilibrium (MPE) of a noncooperative game between sequential---``current'' and ``future''---social planners \citep{klein08}. We will focus on the concept of MPE due to its quantitative tractability, following \cite{bianchi16} and \cite{bianchi18} who applied this approach in the analysis of optimal macroprudential policy in small open economies. Other concepts of time-consistent policies exist, such as sustainable policies \citep{chari90}, and Markov policies are generally inferior to history-contingent sustainable policies.

Denote as $\mathcal{X}\subseteq\mathbb{R}^2$ and $\mathcal{Z}\subseteq\mathbb{R}^2$ the spaces of endogenous and exogenous state variables $X\in\mathcal{X}$ and $z\in\mathcal{Z}$. We have $X=(D,K)$ and $z=(A,\xi)$. Let $\mathcal{S}\equiv\mathcal{X}\times\mathcal{Z}$ with $S\in\mathcal{S}$. To simplify notation, we will often use the subscripts $S$ and $S'$ to denote the values of functions evaluated at those states. Denote the future planner's value and policy functions as $\Bar{V}^h$, $\Bar{C}$, $\Bar{K}'$, $\Bar{L}$, $\Bar{V}^1$, where all functions map $\mathcal{S}\to\mathbb{R}$. Since the current planner can affect the future bank value of survivors $\Bar{V}^1$ only indirectly by affecting $S'=(D',K',z')$, we can use \eqref{eq: aggregate bank value} to solve for $V$, removing it from the set of implementability conditions. The current planner's best response to the future planner's decisions is represented by
\begin{equation*}
    V^h(S)=\max_{(C,D',K',L)\in\mathcal{G}(S)}{U(C,L)+\beta\mathbb{E}_{z}(\Bar{V}^h(S'))},
\end{equation*}
where $\mathcal{G}:\mathcal{S}\to\mathcal{P}(\mathbb{R}^4_{+})$ is defined by the constraints
\begin{align*}
    0&=\sigma(X_{S}K-D)+\Bar{N}+Q(K,K',\xi)(\omega{}K-K')+\beta\mathbb{E}_{z}\left(\frac{U_C(\Bar{C}_{S'},\Bar{L}_{S'})}{U_C(C,L)}\right)D',\\
    0&\le{}\beta\mathbb{E}_{z}\left\{\frac{U_C(\Bar{C}_{S'},\Bar{L}_{S'})}{U_C(C,L)}[(1-\sigma)(X_{S'}K'-D')+\Bar{V}^1_{S'}]\right\}-\theta{}Q(K,K',\xi)K',\\
    0&=U_C(C,L)A{}F_{L}(\xi{}K,L)+U_L(C,L),\\
    0&=A{}F(\xi{}K,L)-C-I(K,K',\xi),
\end{align*}
where
\begin{gather*}
    X_{S}\equiv[A{}F_{K}(\xi{}K,L)+Q(K,K',\xi)(1-\delta)]\xi,\\
    X_{S'}\equiv[A'F_{K}(\xi'K',\Bar{L}_{S'})+Q(K',\Bar{K}'_{S'},\xi')(1-\delta)]\xi'.
\end{gather*}
In an MPE for a given distribution $\Delta:\mathcal{S}\to[0,1)$, $V^h\equiv\Bar{V}^h$ solves the Bellman equation, and policy functions of the current and future planners coincide. In particular, $V^1$ satisfies
\begin{equation*}
    V^1_{S}=\Delta_{S}\mathbb{E}_{z}\{\Lambda_{S,S'}[(1-\sigma)(X_{S'}K'_{S}-D'_{S})+V^1_{S'}]\}.
\end{equation*}

Consistent with the notation used so far, let us denote the Lagrange multipliers on the planner's constraints as $\widetilde{\nu}$, $\widetilde{\lambda}$, $\lambda^L$, and $\lambda^Y$, respectively. Define $\nu_{S}\equiv\frac{\widetilde{\nu}_{S}}{U_{C,S}}$ and $\lambda_{S}\equiv\frac{\widetilde{\lambda}_{S}}{U_{C,S}}$. The next proposition parallels Proposition \ref{prop: CEA characterization} in the context of the MPE.
\begin{proposition}\label{prop: MCEA characterization}
    The CE allocation is generally inefficient compared to the MCEA. Under the assumption of differentiability of the policy functions, the MCEA generalized Euler equations associated with $D'$ and $K'$ can be represented as
    \begin{align*}
        \nu_{S}&=R_{S}\mathbb{E}_{z}\{\Lambda_{S,S'}[(1-\sigma)\lambda_{S}+\sigma\nu_{S'}]\}\underbrace{-\frac{R_{S}\Xi^D_{S}}{U_{C,S}}}_\text{future policy},\\
        \theta\lambda_{S}+\nu_{S}&=\mathbb{E}_{z}\left\{\Lambda_{S,S'}[(1-\sigma)\lambda_{S}+\sigma\nu_{S'}]\frac{X_{S'}}{Q_{S}}\right\}+\Omega^K_{S}+\underbrace{\frac{\Xi^K_{S}}{Q_{S}U_{C,S}}}_\text{future policy},
    \end{align*}
    where for $x\in\{D,K\}$,
    \begin{multline*}
        \Xi^X_{S}\equiv\beta\nu_{S}\mathbb{E}_{z}(\underbrace{U_{CC,S'}\Bar{C}_{X,S'}+U_{CL,S'}\Bar{L}_{X,S'}}_{\text{SDF in deposit rate}\;(-^*)})D'_{S}\\
        +\beta\lambda_{S}\mathbb{E}_{z}\Bigl(\Bigl[(\underbrace{U_{CC,S'}\Bar{C}_{X,S'}+U_{CL,S'}\Bar{L}_{X,S'}}_{\text{SDF in aggregate bank value}\;(-^*)})[(1-\sigma)(X_{S'}K'_{S}-D'_{S})+\Bar{V}^1_{S'}]\\
        +U_{C,S'}\{\underbrace{(1-\sigma)[A'F_{KL,S'}\Bar{L}_{X,S'}+Q_{2,S'}\Bar{K}'_{X,S'}(1-\delta)]\xi'K'_{S}+\Bar{V}^1_{X,S'}}_\text{future asset return and bank value of survivors}\}\Bigr]\Bigr),
    \end{multline*}
    where SDF is the stochastic discount factor, is the combined marginal effect of $X'$ on the current planner's Lagrangian through the policy functions of the future planner $\Bar{C}$, $\Bar{L}$, $\Bar{K}'$, and $\Bar{V}^1$. The capital wedge satisfies
    \begin{align*}
        Q_{S}\Omega^K_{S}&\equiv\underbrace{\nu_{S}Q_{2,S}\{[\sigma(1-\delta)\xi+\omega]K-K'_{S}\}}_{\text{balance sheet}\;(-^*)}\underbrace{-\lambda_{S}\theta{}Q_{2,S}K'_{S}}_{\text{value of default}\;(-)}\underbrace{-\frac{\lambda^Y_{S}}{U_{C,S}}I_{2,S}}_{\text{consumption}\;(-^*)}\\
        &\quad+\underbrace{\mathbb{E}_{z}\{\Lambda_{S,S'}[(1-\sigma)\lambda_{S}+\sigma\nu_{S'}][A'{}F_{KK,S'}\xi'+Q_{1,S'}(1-\delta)]\xi'{}K'_{S}\}}_{\text{future asset return}\;(-)}\\
        &\quad+\mathbb{E}_{z}\Biggl(\Lambda_{S,S'}\Biggl\{\underbrace{\nu_{S'}[\omega(Q_{1,S'}K'_{S}+Q_{S'})-Q_{1,S'}K'_{S'}]}_{\text{future balance sheet}\;(+^*)}\underbrace{-\lambda_{S'}\theta{}Q_{1,S'}K'_{S'}}_{\text{future value of default}\;(+)}+\underbrace{\lambda^L_{S'}A'{}F_{KL,S'}\xi'}_{\text{future wage}\;(+^*)}\\
        &\quad+\underbrace{\frac{\lambda^Y_{S'}}{U_{C,S'}}(A'{}F_{K,S'}\xi'-I_{1,S'})}_{\text{future consumption}\;(+^*)}\Biggr\}\Biggr).
    \end{align*}
\end{proposition}
\begin{proof}
    See Appendix \ref{sec: proof MCEA characterization}.
\end{proof}
First, as in the case of commitment, we must be aware that the planner's (transformed) Lagrange multipliers are generally different from those in the CE. Moreover, the direct quantity effects in the planner's Euler equations (right-hand sides without the wedges) are different from those in the CE: the planner's $(1-\sigma)\lambda_{S}+\sigma\nu_{S'}$ corresponds to the individual banker's $(1+\lambda_{S})(1-\sigma+\sigma\nu_{S'})$, which both reflect the direct effects on the future net worth and the (relevant) continuation value. In the individual banker's problem, the bank value appears in the EC and in the objective function---hence, the multiplication by $1+\lambda_{S}$. Moreover, the shadow value of net worth $\nu$ is linked to the derivative of the banker's value function $v$. In the planner's problem, the objective is the household welfare, so the bank value appears once in the EC (multiplication by $\lambda_{S}$ only). Moreover, the shadow value of net worth is linked to the derivatives of the household value function $V^h$, not being related to the EC---therefore, there is no multiplication by $1+\lambda_{S}$.

Now consider the wedges. Without commitment, the current planner must take into account how its current decisions affect the future endogenous state and the decisions of the future planner, which introduces the $\Xi^D_{S}$ and $\Xi^K_{S}$ terms reflecting those effects. These objects have a symmetric structure, capturing three main transmission mechanisms. First, $D'$ and $K'$ affect the future consumption $\Bar{C}$ and labor $\Bar{L}$ decisions and thus the future marginal utility of consumption and the stochastic discount factor, which affects the deposit rate according to the household Euler equation \eqref{eq: household Euler deposits}. Second, there is a similar effect on the stochastic discount factor implicit in the forward-looking bank value \eqref{eq: aggregate bank value}. Third, $D'$ and $K'$ affect the future net worth at exit and the future bank value of survivors $\Bar{V}^1$ conditional on survival, where the former is generated by the impact on both the future marginal product of capital through $\Bar{L}$ and the future asset price through $\Bar{K}$. Intuitively, we can expect that the derivatives of the policy functions with respect to $K$ are generally nonnegative since greater $K$ is associated with both greater output and a greater bank net worth. On the contrary, a greater bank debt $D$ has a negative effect on net worth, investment, and the household value function, so we can expect that the derivatives of the policy functions are generally nonpositive. The combined effects and the signs of $\Xi^D_{S}$ and $\Xi^K_{S}$ remain ambiguous.\footnote{Our quantitative approach is to find a fixed point in the Bellman equation and the policy functions directly instead of solving the KKT conditions, so we will not be assuming that the policy functions are differentiable.}

The additional capital wedge $\Omega^K_{S}$ corresponds to a similar term arising under commitment. Contrary to the latter, the time-consistent planner cannot affect the $t-1$ expectations of the asset return and the bank value distribution at $t$. Likewise, without commitment, the planner cannot affect the future distribution except for the indirect impact through the future endogenous states. For this reason, we did not make the distribution explicit in the continuation value of survivors $\Bar{V}^1$. The remaining effects---the negative balance sheet, value of default, and consumption channels, the corresponding positive future effects, and the negative impact on the future asset return---are identical to the case of commitment. A quantitative exploration is generally required to assess which effects dominate. Indeed, as we will see, the combined effect is typically not uniformly positive or negative but state-contingent, allowing for both excessive and insufficient borrowing and lending in the CE.

Unlike in the case of commitment, we do not have a formal statement on the welfare ranking of Markov perfect outcomes corresponding to different $\Delta$. A shift in $\Delta$ directly affects the fixed point as we iterate on $\Bar{V}^1$, so the welfare effects may have different signs in different regions of the state space. We can, however, state with certainty that a uniform positive shift in $\Delta$ must increase welfare in the steady state in which the EC is binding.

\subsection{Implementation with taxes, transfers, and capital requirements}
The presence of two broad sources of inefficiencies---various pecuniary externalities and a potentially suboptimal bank value distribution---generally requires two types of policy instruments to implement the CEA (MCEA) in a regulated CE. A given distribution $\Delta$ can naturally be achieved with entrants/survivors-specific transfers within the banking system. The wedges in the Euler equations can be addressed with proportional taxes on bank deposits and assets or, under some assumptions, with state-contingent capital requirements. 
The next proposition formalizes the alternative ways of implementing the CEA (MCEA) in a regulated CE. We will use the sequential (CEA) notation where $\{x_{t}\}$ denotes a sequence of functions $x_{t}:S^{t}\to\mathbb{R}$, while the implicit recursive (MCEA) analog is a single function $x:\mathcal{S}\to\mathbb{R}$.
\begin{proposition}\label{prop: implementation}
    Consider a regulated CE with balance sheet taxes, such that a banker $i\in[0,f]$ has the budget constraint
    \begin{equation*}
        (1+\tau^K_{t})Q_{t}k_{i,t+1}=(1+\tau^{j(i)}_{t})n_{i,t}+(1-\tau^D_{t})\frac{d_{i,t+1}}{R_{t}},
    \end{equation*}
    where
    \begin{equation*}
        j(i)=
        \begin{cases}
            1 & i\in[0,\sigma{}f]\\
            0 & i\in(\sigma{}f,f]
        \end{cases}.
    \end{equation*}
    Moreover, the banker faces a regulatory constraint
    \begin{equation*}
        (1+\tau^{j(i)}_{t})n_{i,t}\ge\kappa_{t}(1+\tau^K_{t})Q_{t}k_{i,t+1}.
    \end{equation*}
    The government budget constraint is
    \begin{equation*}
        \tau^D_{t}\frac{D_{t+1}}{R_{t}}+\tau^K_{t}Q_{t}K_{t+1}=\tau^1_{t}N^1_{t}+\tau^0_{t}N^0_{t}.
    \end{equation*}
    The CEA (MCEA) can be implemented with $(\tau^D_{t},\tau^0_{t},\tau^1_{t})$ or $(\tau^K_{t},\tau^0_{t},\tau^1_{t})$, such that $\tau^D_{t}$ or $\tau^K_{t}$ internalizes the pecuniary externalities, and $(\tau^0_{t},\tau^1_{t})$ implement the targeted bank value distribution. The minimum capital requirement $\kappa_{t}$ is not, in general, a substitute for $\tau^D_{t}$ or $\tau^K_{t}$. The CEA (MCEA) and the policy that implements it constitute a Ramsey (Markov perfect) equilibrium.
\end{proposition}
\begin{proof}
    See Appendix \ref{sec: proof implementation}.
\end{proof}

As explained in Appendix \ref{sec: proof implementation}, we construct all the policies using the primal approach. The asset or deposit tax closes the wedges, the net worth subsidy to old banks targets the bank value distribution, and the subsidy to entrants balances the government budget. The regulatory capital requirement alone is sufficient to account for the wedges if and only if a measure of a discounted credit spread stays nonnegative in the CEA (MCEA). A sufficient condition for the latter is that $\mathbb{E}_{t}[\Lambda_{t,t+1}f_{t+1}(\frac{X_{t+1}}{Q_{t}}-R_{t})]\ge{}0$ for all positive-valued $f_{t+1}$. A necessary and sufficient condition is that it holds for $f_{t+1}=1-\sigma+\sigma(\Tilde{\nu}_{t+1}+\Bar{\xi}_{t+1})$, where $\Bar{\xi}_{t}$ is a transformation of the Lagrange multiplier on the regulatory constraint. A difficulty is that $\{\Tilde{\nu}_{t},\Bar{\xi_{t}}\}$ solve a system of two stochastic difference equations (the banker's deposit and asset Euler equations) conditional on the optimal allocation. Quantitatively, the required condition does not always hold: the planner can optimally choose to have a negative discounted credit spread in some contingencies. In this case, the capital requirement alone fails to be effective.

\section{Quantitative results}\label{sec: quantitative results}
This section describes the model calibration and conducts a multifaceted comparison of the CE, MCEA, and CEA allocations. We will investigate the efficiency of borrowing and lending by the banking system, explore the properties of optimal policies, analyze welfare gains, compare the economic dynamics around financial crises, and study the implications of alternative bailout policies.

\subsection{Calibration and computation}
I assume separable preferences for households:
\begin{equation*}
    U(C,L)=\lim_{x\to\gamma}\frac{C^{1-x}-1}{1-x}-\chi\frac{L^{1+\phi}}{1+\phi},
\end{equation*}
where $\gamma,\phi,\chi\ge0$. The final and capital good production technologies are $F(\xi{}K,L)=(\xi{}K)^\alpha{}L^{1-\alpha}$ with $\alpha\in(0,1)$, and $\Phi(x)=\zeta+\kappa_1{}x^\psi$ with $\zeta\in\mathbb{R}$, $\kappa_1>0$, and $\psi\in(0,1]$. The logs of exogenous stochastic processes $\{A_{t}\}$ and $\{\xi_{t}\}$ are AR(1) with autocorrelations $(\rho_a,\rho_\xi)$ and standard deviations $(\sigma_a,\sigma_\xi)$.

Table \ref{tab: parameters} reports the parameter values that are mostly set to reflect long-run facts about the US economy in 1990--2019.
\begin{table}[ht!]
    \caption{Parameter values}\label{tab: parameters}
    \centering
    \begin{tabular*}{\textwidth}{l@{\extracolsep{\fill}}ll}
        \toprule
        Parameter & Value & Target\\
        \midrule
        \multicolumn{3}{l}{Preferences and technology}\\
        \midrule
        $\alpha$ & 0.404 & labor share $\approx{59.6}\%$ \\
        $\beta$ & 0.995 & annualized real interest rate = 2\% \\
        $\gamma$ & 1 & log preferences from consumption \\
        $\delta$ & 0.02 & annual depreciation rate $\approx{7.6}\%$ \\
        $\zeta$ & -0.007  & $\frac{I}{K}=\delta$ \\
        $\kappa_1$ & 0.499 & $Q=1$ \\
        $\phi$ & 0.625 & microfounded aggregate Frisch elasticity = 1.6 \\
        $\chi$ & 0.86 & $L=1$ \\
        $\psi$ & 0.75 & panel data estimates in the literature \\
        \midrule
        \multicolumn{3}{l}{Banking}\\
        \midrule
        $\Bar{N}=0$ & 0 & linear endowment rule \\
        $\sigma$ & 0.976 & bank exit probability $\approx{0.091}$ \\
        $\theta$ & 0.216 & \multirow{2}{*}{$N/(Q{}K)=0.125$, annualized credit spread = 0.5\%} \\
        $\omega$ & 0.001 & \\
        \midrule
        \multicolumn{3}{l}{Exogenous stochastic processes}\\
        \midrule
        $\rho_a$ & 0.935 & \multirow{4}{*}{\begin{minipage}{0.5\textwidth}$\corr(\widehat{Y}_{t},\widehat{Y}_{t-1})\approx{0.886}$, $\corr(\widehat{I}_{t},\widehat{I}_{t-1})\approx{0.894}$,  $\sd(\widehat{Y}_{t})\approx{0.013}$, $\sd(\widehat{I}_{t})\approx{0.045}$\end{minipage}} \\
        $\rho_\xi$ & 0.956 & \\
        $\sigma_a$ & 0.006 & \\
        $\sigma_\xi$ & 0.002 & \\
        \bottomrule
    \end{tabular*}
    \begin{tabular}{@{}p{\textwidth}@{}}
        {\small Notes: $\widehat{X}_{t}$ denotes the cyclical component of $\ln(X_{t})$ extracted using the HP filter with $\lambda=1600$.}
    \end{tabular}
\end{table}
The Cobb---Douglas elasticity $\alpha$ targets the average labor share in the nonfarm business sector based on the US Bureau of Labor Statistics data. The discount factor $\beta$ corresponds to the annualized real interest rate of 2\%. The risk-aversion $\gamma$ is set to unity, implying log preferences from consumption, as common in the literature. The depreciation rate $\delta$ proxies the average depreciation rate of the current-cost net stock of private fixed assets and consumer durables in the Bureau of Economic Analysis data. The capital production technology parameters $(\zeta,\kappa_1)$ are set to have $\frac{I}{K}=\delta$ and normalize $Q=1$ in the deterministic steady state, while $\psi$ is set as in \cite{gertler20RESTUD} to match panel data estimates. The labor disutility elasticity $\phi$---an inverse of the Frisch elasticity of labor supply--- targets the average of the microfounded estimates of the aggregate Frisch elasticity for males \citep{erosa16} and females \citep{attanasio18}. The labor disutility scale $\chi$ corresponds to a normalization $L=1$ in the steady state.

It is computationally convenient to set $\Bar{N}=0$ so that the aggregate endowment of entrants is linear in the assets of exiting bankers. I set the survival probability $\sigma$ based on the average establishment exit rate in finance, insurance, and real estate according to the Business Dynamics Statistics data. The remaining banking parameters $(\theta,\omega)$ target the average capital ratio of 12.5\%---consistent with the evidence in \cite{begenau20} that for most banks, regulatory constraints are not binding---together with the annualized credit spread of 0.5\% so that the EC binds in the CE less than half of the time.

The AR(1) parameters $(\rho_a,\rho_\xi)$ and $(\sigma_a,\sigma_\xi)$ target the autocorrelations and standard deviations of output and investment, using the National Income and Product Accounts data. Each variable is logged and detrended using the HP filter with $\lambda=1600$, a standard value for quarterly data.

To compute the CE and MCEA, I use global projection methods \citep{judd98} so that the nonlinearity due to the occasionally binding EC can be fully addressed. Specifically, I approximate the CE and MCEA unknown functions with linear 2D splines for each $z\in\widehat{\mathcal{Z}}\subset\mathcal{Z}$. (Accordingly, I approximate the exogenous stochastic process $\{A_{t},\xi_{t}\}$ by a finite-state Markov chain $z\mapsto{}z'$.) In the case of the CEA, I employ both the global projection method---linear 4D splines---and the local piecewise linear perturbation method \citep{guerrieri15} that respects occasionally binding constraints but not precautionary savings. The latter method serves as the baseline, but I verify some results with the global method on a coarse grid. Since Lagrange multipliers $\gamma_{t-1}$ and $\nu_{t-1}$ must be treated as state variables, the complexity of the Ramsey problem combined with the course of dimensionality makes fully global approximation challenging.

Instead of the natural endogenous state $(D,K)$, I work with a rotated state space based on $(\log(D),\log(K))$. This way we can account for the strong positive correlation between $\log(D)$ and $\log(K)$, which is illustrated in Figure \ref{fig: state space} in the CE case.
\begin{figure}[ht!]
    \caption{Endogenous state space, CE ergodic distribution}\label{fig: state space}
    \centering
    \includegraphics[trim={1.8cm 6.8cm 1.8cm 7.3cm},clip,width=0.75\textwidth]{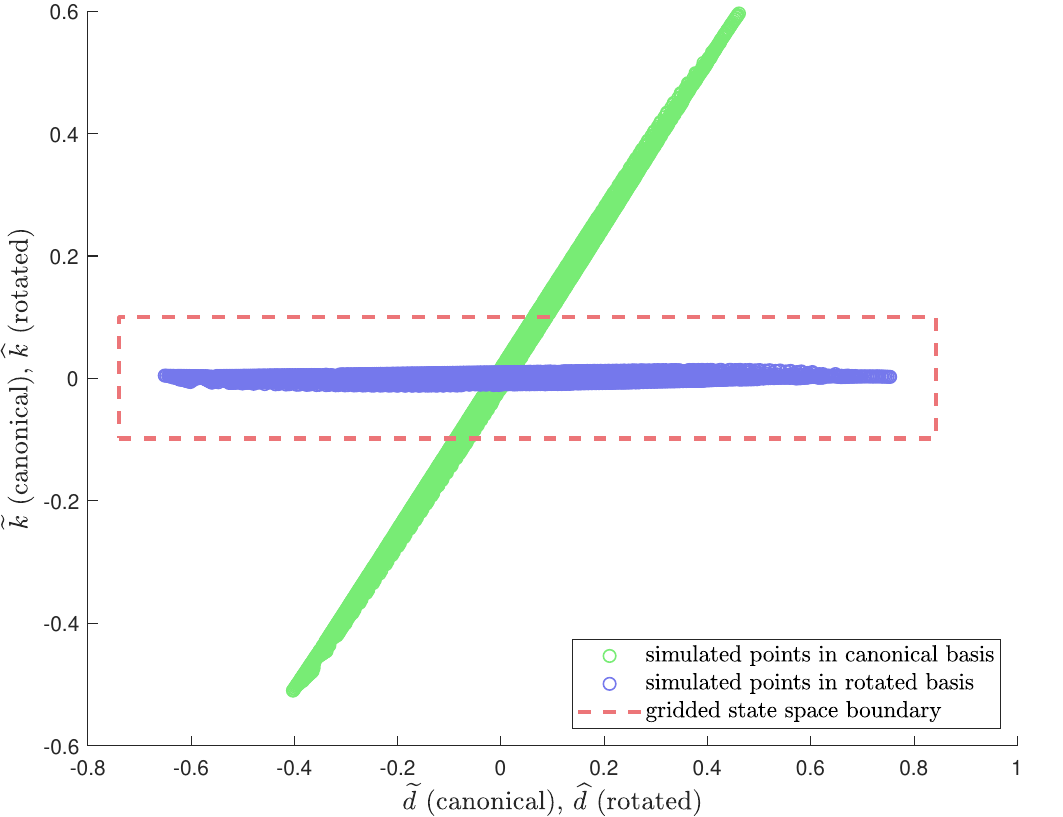}
    \begin{tabular}{@{}p{\textwidth}@{}}
        {\small Notes: For $x\in\{D,K\}$, $\widetilde{x}\equiv{}\log(X)-\widehat{\mathbb{E}}(\log(X))$, and $(\widehat{d},\widehat{k})$ are obtained by rotating $(\widetilde{d},\widetilde{k})$ clockwise at the angle $\arctan\left(\frac{\widehat{\cov}(\widetilde{d},\widetilde{k})}{\widehat{\vvar}(\widetilde{d})}\right)$.}
    \end{tabular}
\end{figure}

\subsection{Bank solvency and enforcement constraint regimes}
As illustrated in Figure \ref{fig: bank solvency and EC}, the endogenous state space can be divided into three regions: banks are solvent and unconstrained (highlighted in yellow), solvent but constrained (light green), and insolvent and constrained (dark green).
\begin{figure}[ht!]
    \caption{Bank solvency and EC regimes in the worst and best exogenous states in the CE}\label{fig: bank solvency and EC}
    \centering
    \includegraphics[trim={2.0cm 14.1cm 2.0cm 6.9cm},clip,width=0.75\textwidth]{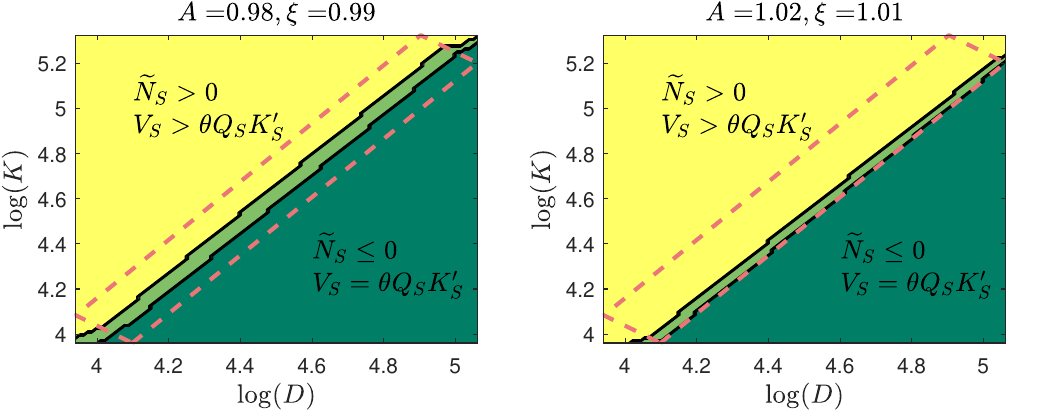}
    \begin{tabular}{@{}p{\textwidth}@{}}
        {\small Notes: In the yellow region, banks are solvent, and the EC has slack. In the light green region, banks are solvent, but the EC is binding. In the dark green region, banks are insolvent, and the EC is binding. The dashed parallelogram (not a rectangle due to scaling) is the boundary of the endogenous state space represented in the canonical basis.}
    \end{tabular}
\end{figure}
Banks cannot be insolvent and unconstrained in the CE simultaneously. If the survived banks are insolvent, their bank value is zero, so the EC must be binding for them. Since the Lagrange multipliers depend only on the aggregate state, the entering banks must also be constrained.

According to Figure \ref{fig: bank solvency and EC}, banks are solvent and unconstrained when the initial capital stock $K$ is sufficiently large compared to bank debt $D$. There generally exist thresholds $\Bar{K}(D,z)$ and $\widehat{K}(D,z)$, such that banks are solvent when $K>\Bar{K}(D,z)$ and are, moreover, unconstrained when $K>\widehat{K}(D,z)\ge\Bar{K}(D,z)$. Based on the figure, we can conjecture that both $\Bar{K}$ and $\widehat{K}$ are decreasing in $z$ in the sense that $\Bar{K}(D,s_2)\le\Bar{K}(D,s_1)$ when $A_2>A_1$ and $\xi_2>\xi_1$. The thresholds are also generally increasing in $D$. An analytic characterization of $\Bar{K}$ and $\widehat{K}$ does not seem possible, but the conjectured properties are intuitive.

Although the area of the insolvency region might seem significant, the model does not typically visit those states. According to Figure \ref{fig: state space}, the ergodic set is a thin ellipse inside the gridded state space (the dashed parallelogram in Figure \ref{fig: bank solvency and EC}). Insolvency is more likely in the worst exogenous state, but it does not typically occur even in that case. On the other hand, the model stays in the binding EC regime approximately 40\% of the time in the CE.

Figure \ref{fig: bank solvency and EC} confirms the potential welfare benefits from preemptive bailouts. By keeping banks away from the solvent-but-constrained buffer zone, the policymaker escapes the potentially harmful effects of being in the constrained regime and decreases the probability of ending up in the insolvency region, at which point the banking system would collapse.

Figure \ref{fig: net bank value slices k CE} further explores how the magnitude of the distance between the aggregate bank value and the value of default $V_{S}-\theta{}Q_{S}K'_{S}$ varies in the state space.
\begin{figure}[ht!]
    \caption{Net bank value in the CE}\label{fig: net bank value slices k CE}
    \centering
    \includegraphics[trim={2.3cm 13.5cm 2.3cm 6.9cm},clip,width=0.75\textwidth]{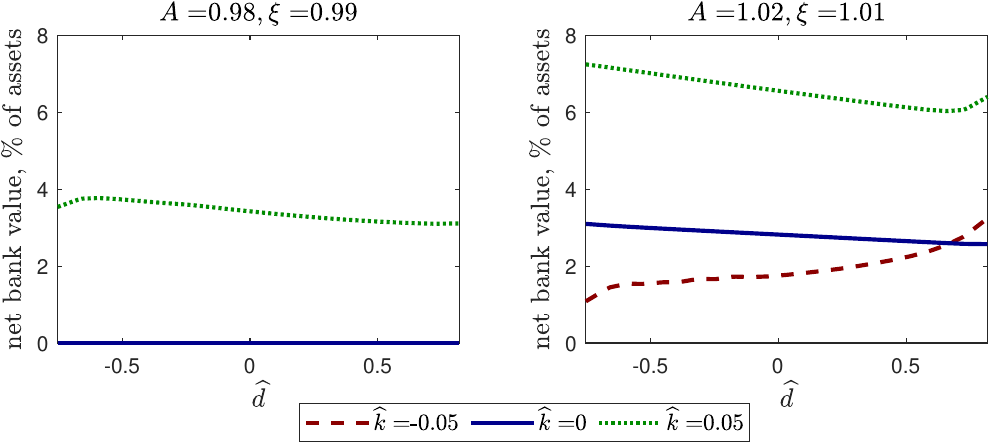}
    \begin{tabular}{@{}p{\textwidth}@{}}
        {\small Notes: Slices of the underlying surfaces along the $\widehat{d}$ dimension at the quartiles of the $\widehat{k}$ grid. The y-axis is $V_{S}-\theta{}Q_{S}K'_{S}$ in \% of $Q_{S}K'_{S}$.}
    \end{tabular}
\end{figure}
We now focus on the gridded state space where the model is solved. Figure \ref{fig: net bank value slices k CE} displays the variation along the $\widehat{d}$ dimension, that is, moving from the southwest to northeast inside the dashed parallelogram in Figure \ref{fig: bank solvency and EC} at the quartiles of the $\widehat{k}$ grid. In the lowest exogenous state, the EC is typically binding at higher leverage ratios, e.g., at or below the median of $\widehat{k}$. The constraint has slack when banks are more capitalized (higher $\widehat{k}$), and the slack in proportion to bank assets slightly decreases as the balance sheet expands (larger $\widehat{d}$). In the highest exogenous state, the expected asset returns are greater and financial constraints are mostly nonbinding, especially at the higher capital ratios. As with the lowest state, the relative slack generally decreases as the balance sheet expands at higher capital ratios; however, there is an opposite relationship when banks are more leveraged. These regularities indicate that a way to improve over the CE is to relax the binding ECs when exogenous conditions are worse.

\subsection{Financial crises in the unregulated economy}
When the banker's EC binds, a banker is indifferent between continuing to run the banking business and defaulting on liabilities and running away with a fraction of assets. Our convention is that bankers continue to operate at the point of indifference. The instances where the aggregate EC is about to switch from having slack to being binding and the ensuing spells in the binding regime with the associated deleveraging can naturally correspond to the build-up of systemic risk and financial crises. The risk is systemic because our banks make symmetric decisions: when the EC binds for one bank, it binds for all. This subsection explores the economic dynamics around such episodes.

Define a financial crisis that starts at $t$ as an event that satisfies two conditions on the behavior of the aggregate EC: it has slack for at least five years before the crisis ($[t-20,t-1]$) and then binding for at least one year ($[t,t+3]$). Figure \ref{fig: financial crises CE} illustrates typical dynamics around such crises.
\begin{figure}[ht!]
    \caption{Financial crises in the CE}\label{fig: financial crises CE}
    \centering
    \includegraphics[trim={2.2cm 11.8cm 2.2cm 7.1cm},clip,width=0.75\textwidth]{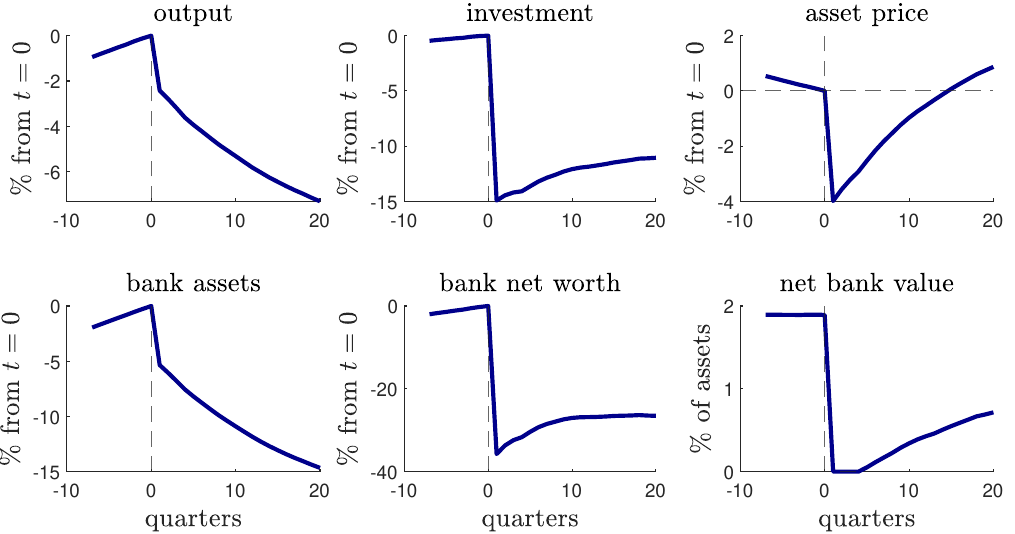}
    \begin{tabular}{@{}p{\textwidth}@{}}
        {\small Notes: Averages over a 1,000,000-period simulation.}
    \end{tabular}
\end{figure}
The figure is obtained by simulating the CE for 1,000,000 periods (quarters), selecting crisis episodes as defined above, and averaging the simulated paths. There are 8,106 such crises, which corresponds to approximately 3.2 financial crises per century, consistent with the findings in the related literature \citep{mendoza10}.

Financial crises have a boom-bust pattern. Ahead of a crisis, output, consumption, and investment are increasing, and the balance sheet of the banking system is expanding. A leading indicator of the crisis is the gradually falling forward-looking asset price. The aggregate EC binds when a bad exogenous state occurs, typically due to a decrease in capital quality. The asset price and the realized return on bank assets drop, which triggers a sharp fall in bank net worth---the bust starts. As banks deleverage, balance sheets shrink, firms cut investment, and an economic recession starts. There is a slight rise in consumption on impact due to the fall in the deposit rate and the increase in labor supply, but the effect is short-term, as consumption starts to fall next period. Meanwhile, the forward-looking asset price starts to recover, and so does bank net worth and the aggregate investment. As bank deleveraging continues, the EC switches to having slack again, and the bank value slowly begins to recover. The fall in output gradually slows down, but the recession and financial deleveraging persist.

\subsection{Ramsey equilibrium}
In this subsection, we will explore the implication of the Ramsey equilibrium, relating them to the findings discussed so far. The Ramsey allocation is not recursive but history-dependent; therefore, we cannot directly compare policy functions with those in the CE.\footnote{Since the Ramsey equilibrium is recursive on the state space augmented with Lagrange multipliers, it is possible to compare policy functions conditional on specific values of Lagrange multipliers.} We will focus on comparing the empirical distributions and the economic dynamics around financial crises. The computation of the CEA relies on piecewise linear perturbation about the steady state. For consistency, in this subsection, we use the same method to compute the CE. The computational burden of simulating the model using this approach is significant, so the simulation length is reduced from 1,000,000 to 100,000.

The optimal bank value distribution among constant distributions $\Delta_{t}(s^{t})=\Delta$ for all $(t,s^{t})$ is $\Delta\approx{}0.9985$---the smallest value at which the aggregate EC has slack in the steady state. The ergodic welfare gain from the CEA conditional on the optimal distribution $\Delta$ is about 0.75\% of consumption. Note that the steady-state value of the CE distribution is $\Bar{\sigma}^1\approx{}0.9911$. The Ramsey planner finds it optimal to promise sufficiently large transfers to the banks, such that the EC becomes just nonbinding in the steady state---that is, banks are exactly at the boundary of the constrained and unconstrained regions.

We begin by looking at the empirical distributions of bank deposits and loans. Figure \ref{fig: histograms deposits and loans CEA} shows the corresponding histograms, where in addition to the CE and CEA, we have histograms from the frictionless (unconstrained) CE (UE), in which $\theta=0$ and all other parameters are identical to those in the baseline CE. By construction, in the UE, the EC is always nonbinding since the aggregate net worth and bank value are strictly positive.
\begin{figure}[ht!]
    \caption{Bank borrowing and lending in the CE, CEA, and UE}\label{fig: histograms deposits and loans CEA}
    \centering
    \includegraphics[trim={2.4cm 13.5cm 2.4cm 7.3cm},clip,width=0.75\textwidth]{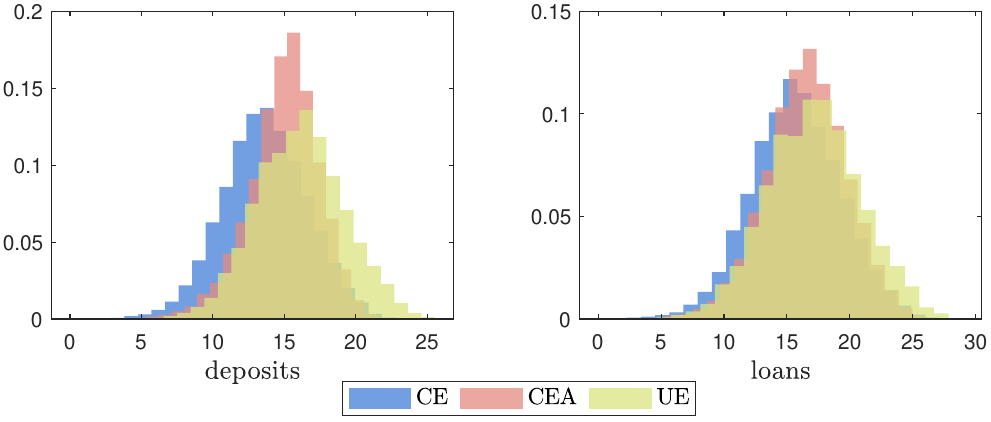}
    \begin{tabular}{@{}p{\textwidth}@{}}
        {\small Notes: UE refers to a frictionless (unconstrained) CE with $\theta=0$. Histograms based on the 100,000-period simulation with the same sequence of exogenous shocks. Variables are normalized by the average CE output; the y-axis has the pdf normalization.}
    \end{tabular}
\end{figure}

An immediate implication of the optimal preemptive bailout policy that supports the relative bank value of survived banks at a greater value than in the CE is the expansion of bank balance sheets. With commitment, we observe \emph{underborrowing} and \emph{underlending} by the banking sector in the CE compared to the CEA. Bank deposits and loans have a greater mean and variance in the Ramsey equilibrium than in the CE. The CEA histograms are more skewed to the left, so the median deposits and loans are even greater. However, the optimal balance sheets are smaller than in the UE, where the EC is always nonbinding and the limited enforcement friction is shut down. Hence, the Ramsey planner alleviates the friction with preemptive bailouts but cannot eliminate it completely.

Figure \ref{fig: histograms policies CEA} displays the empirical distributions of the alternative CEA implementation policies.
\begin{figure}[ht!]
    \caption{Optimal policies under commitment}\label{fig: histograms policies CEA}
    \centering
    \includegraphics[trim={2.5cm 11.3cm 2.5cm 7.1cm},clip,width=0.75\textwidth]{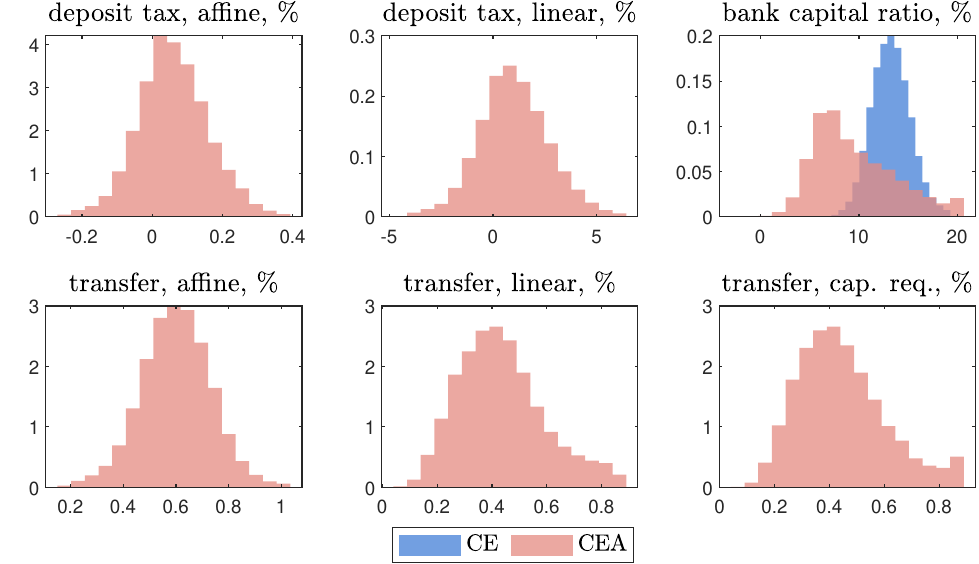}
    \begin{tabular}{@{}p{\textwidth}@{}}
        {\small Notes: Histograms based on the 100,000-period simulation with the same sequence of exogenous shocks. Each column corresponds to an alternative CEA decentralization scheme (pdf normalization on the y-axis). Outliers are removed. Transfers ($\tau^1_{t}$) are in \% of bank assets. The transfer in the last column is meaningful only when the implied Lagrange multiplier on the regulatory constraint is nonnegative, which does not always hold.}
    \end{tabular}
\end{figure}
Although there is greater bank borrowing and lending in the CEA than in the CE, the optimal deposit taxes have more mass in the positive region. The Ramsey planner uses taxes to correct the pecuniary externalities, which prevents borrowing and lending from being excessively large, even though it is larger than in the CE due to optimal transfers. When aggregate transfers are forbidden (the linear implementation scheme), the magnitude of the taxes is generally greater.

The optimal bank capital ratio has a lower mean and median than in the CE but a greater variance and a much greater skewness to the right. Hence, under commitment, the optimal capital ratios are generally lower than in the CE, reflecting the increased borrowing and lending, but there is a nontrivial measure of contingencies in which the planner finds it optimal for banks to be sufficiently more capitalized than in the CE.

The optimal transfers to survived banks are uniformly positive independently of the CEA implementation mechanism. The optimal transfers have a comparable magnitude across implementation schemes with a mean of about 0.5\% of bank assets. The decentralization with capital requirements alone does not always succeed; therefore, the optimal transfers are only valid conditional on having the Lagrange multiplier on the regulatory constraint nonnegative in the relevant contingencies.

Finally, we will look at the economic dynamics around financial crises. Remember that in this subsection, we use a different approach to compute the CE for consistency with the computation of the CEA. Our identification of financial crises changes slightly: instead of requiring the EC to have slack for twenty quarters before the crisis, we look for at least ten quarters, which allows obtaining a similar frequency of financial crises of about 3.1 crises per century.

Figure \ref{fig: financial crises CE and CEA} illustrates the dynamics around crises.
\begin{figure}[ht!]
    \caption{Financial crises in the CE and CEA}\label{fig: financial crises CE and CEA}
    \centering
    \includegraphics[trim={2.3cm 10.9cm 2.3cm 7.1cm},clip,width=0.75\textwidth]{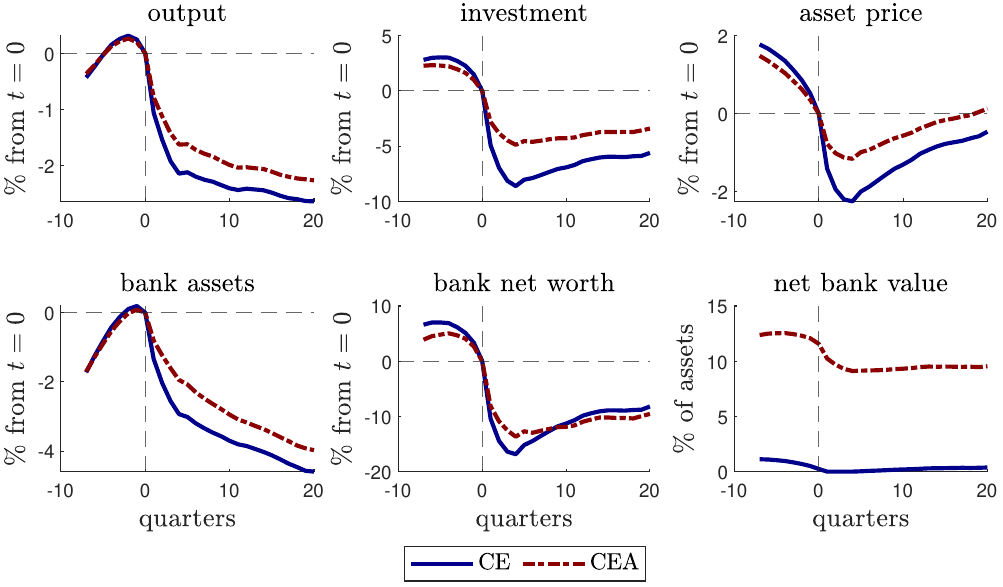}
    \begin{tabular}{@{}p{\textwidth}@{}}
        {\small Notes: Averages over a 100,000-period simulation.}
    \end{tabular}
\end{figure}
The optimal bank capital ratio is uniformly lower during the crises than in the CE. These dynamics reflect the comparison of empirical distributions in Figure \ref{fig: histograms policies CEA}. Since the planner finds it optimal to provide sufficient support to survivors through transfers, they generally borrow more and become more leveraged on average. A lower capital ratio is not a problem since the very purpose of those transfers---or preemptive bailouts---is to prevent the EC from switching to the binding regime, which is achieved successfully---the net bank value in the CEA generally remains nonbinding around crises.

The boom-bust dynamics in the CEA are generally less pronounced than in the CE, as both real and financial variables are less volatile in such episodes and recover faster after the bad shock hits. In particular, we observe a faster recovery in the asset price and bank assets and liabilities, which does not allow investment to drop as severely as in the CE. Consumption also varies less, and output rebounds faster.

Figure \ref{fig: financial crises policies CEA} focuses on the dynamics of the CEA decentralization policies.
\begin{figure}[ht!]
    \caption{CEA decentralization policies around financial crises}\label{fig: financial crises policies CEA}
    \centering
    \includegraphics[trim={2.1cm 11.8cm 2.1cm 7.1cm},clip,width=0.75\textwidth]{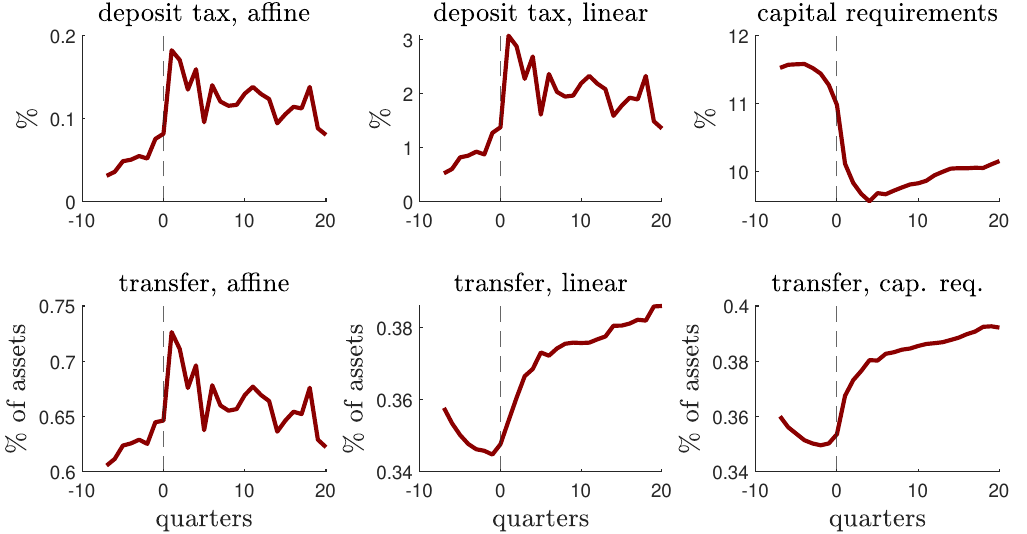}
    \begin{tabular}{@{}p{\textwidth}@{}}
        {\small Notes: Averages over a 100,000-period simulation. Each column corresponds to an alternative implementation mechanism. Transfers refer to $\tau^1$.}
    \end{tabular}
\end{figure}
The optimal deposit taxes are increasing ahead of a crisis and jump when the bad shock arrives to encourage faster deleveraging and keep the banking sector in the unconstrained regime. The increase is followed by a gradual decline as both exogenous and endogenous conditions improve.

By construction, the optimal transfer in the affine scheme tracks the dynamics of the deposit tax to a great extent. In the linear scheme, the optimal transfer is falling slightly ahead of a crisis, which is an additional way to encourage deleveraging ex ante. When the shock arrives, the trend is reversed, and the transfer increases to relax the EC. We observe very similar behavior in the optimal transfer conditional on the implementation with capital requirements. The optimal transfers in all implementation schemes are generally positive around crises.

The implementation with capital requirements is effective around crises since the implied Lagrange multiplier on the regulatory constraint stays positive. At the same time, the optimal capital ratio is generally lower than in the CE. It might seem unintuitive, but remember that the optimal constant bank value distribution scales the CE analog up in the Ramsey equilibrium, so the corresponding transfers would decrease the CE capital ratio to even lower values in the environment without additional regulation. Therefore, with optimal capital requirements and preemptive bailouts, the regulatory constraint would bind in the regulated CE.

\section{Conclusion}\label{sec: conclusion}
This paper has characterized the optimal regulation of a banking system in a quantitative general equilibrium environment. The optimal policy requires a combination of system-wide asset or deposit taxes---that address pecuniary externalities implicit in the banking system enforcement constraint---and bank entrants/survivors-specific transfers that achieve the optimal bank value distribution. We have referred to the optimal transfers as preemptive bailouts, as their goal is to prevent financial constraints from becoming binding, which guarantees bank solvency.

We have studied the optimal policy in the Markov perfect equilibrium and the Ramsey equilibrium, which differ in whether the policymaker can commit to a history-contingent plan. Independently of the latter, the optimal transfer policy generally ensures that the enforcement constraint has slack in most states/contingencies. The presence of commitment has important quantitative implications, as we mostly observe underborrowing and underlending in competitive markets compared to the Ramsey outcome.

The present analysis can be extended in various ways. We could consider alternative environments in which banks can self-insure with endogenous equity issuance or can invest in other types of assets, such as government debt, which will potentially introduce additional externalities. It is also interesting to generalize the model and explore the implications for optimal monetary policy.

\bibliography{lit}

\appendix

\hypertarget{appendices}{\appendixpage}

\section{Proofs}

\subsection{Proposition \ref{prop: banker's problem}}\label{sec: proof banker's problem}
Consider the problem of a banker that enters the banking business at $t=0$. The banker's Lagrangian can then be written as
\begin{multline*}
    \mathcal{L}=v_{0}+\mathbb{E}_{0}\biggl[\sum_{t=0}^\infty\sigma^{t}\beta^{t}\frac{U_{C,t}}{U_{C,0}}\biggl\{\Tilde{\nu}_{t}\left(n_{t}+\frac{d_{t+1}}{R_{t}}-Q_{t}k_{t+1}\right)\\
    +\gamma_{t}[\mathbb{E}_{t}\{\Lambda_{t,t+1}[(1-\sigma)n_{t+1}+\sigma{}v_{t+1}]\}-v_{t}]+\Tilde{\lambda}_{t}(v_{t}-\theta{}Q_{t}k_{t+1})\biggr\}\biggr],
\end{multline*}
where $n_{t+1}\equiv{}X_{t+1}k_{t+1}-d_{t+1}$. The first-order conditions (FOCs) with respect to $v_{t}$, $d_{t+1}$, and $k_{t+1}$ are
\begin{align}
    \gamma_{t}&=\gamma_{t-1}+\Tilde{\lambda}_{t},\qquad\gamma_{-1}=1,\label{eq: banker FOC value}\\
    \Tilde{\nu}_{t}&=\mathbb{E}_{t}\{\Lambda_{t,t+1}[(1-\sigma)\gamma_{t}+\sigma\Tilde{\nu}_{t+1}]\}R_{t},\label{eq: banker FOC deposit}\\
    \theta\Tilde{\lambda}_{t}+\Tilde{\nu}_{t}&=\mathbb{E}_{t}\{\Lambda_{t,t+1}[(1-\sigma)\gamma_{t}+\sigma\Tilde{\nu}_{t+1}]\frac{X_{t+1}}{Q_{t}}\}.\label{eq: banker FOC asset}
\end{align}
The complementary slackness conditions are
\begin{equation}\label{eq: banker CSC value}
    0=\Tilde{\lambda}_{t}(v_{t}-\theta{}Q_{t}k_{t+1}),\qquad\Tilde{\lambda}_{t}\ge0,\qquad{}v_{t}\ge\theta{}Q_{t}k_{t+1}.
\end{equation}
Note that \eqref{eq: banker FOC value} and \eqref{eq: banker CSC value} imply $\gamma_{t}=1+\sum_{j=0}^{t}\Tilde{\lambda}_{j}\ge1$.

\subsubsection{Stationary transformation}
It follows from \eqref{eq: banker FOC value} and \eqref{eq: banker CSC value} that for any history $s^\infty\in{S^\infty}$, the sequence $\{\gamma_{t}\}_{t=0}^\infty$ is nondecreasing and generally unbounded. Define $\nu_{t}\equiv\frac{\Tilde{\nu}_{t}}{\gamma_{t-1}}$, $\lambda_{t}\equiv\frac{\Tilde{\lambda}_{t}}{\gamma_{t-1}}$, and note that $\frac{\gamma_{t}}{\gamma_{t-1}}\equiv{}1+\lambda_{t}$. Then \eqref{eq: banker FOC deposit} and \eqref{eq: banker FOC asset} are equivalent to
\begin{align}
    \nu_{t}&=(1+\lambda_{t})\mathbb{E}_{t}[\Lambda_{t,t+1}(1-\sigma+\sigma\nu_{t+1})]R_{t},\label{eq: banker Euler deposit stationary}\\
    \theta\lambda_{t}+\nu_{t}&=(1+\lambda_{t})\mathbb{E}_{t}[\Lambda_{t,t+1}(1-\sigma+\sigma\nu_{t+1})\frac{X_{t+1}}{Q_{t}}].\label{eq: banker Euler asset stationary}
\end{align}
By construction, $\Tilde{\lambda}_{t}>0$ if and only if $\lambda_{t}>0$, so \eqref{eq: banker CSC value} is equivalent to
\begin{equation}
    0=\lambda_{t}(v_{t}-\theta{}Q_{t}k_{t+1}),\qquad\lambda_{t}\ge0,\qquad{}v_{t}\ge\theta{}Q_{t}k_{t+1}.\label{eq: banker CSV value stationary}
\end{equation}

\subsubsection{Value function}\label{sec: CE value function}
Define $\mu_{t}\equiv{}v_{t}-\nu_{t}n_{t}$. Then
\begin{align*}
    v_{t}&=\mathbb{E}_{t}\{\Lambda_{t,t+1}[(1-\sigma)n_{t+1}+\sigma{}v_{t+1}]\}\\
    &=\mathbb{E}_{t}\{\Lambda_{t,t+1}(1-\sigma+\sigma\nu_{t+1})(X_{t+1}k_{t+1}-d_{t+1})\}+\sigma\mathbb{E}_{t}(\Lambda_{t,t+1}\mu_{t+1})\\
    &=\frac{\theta\lambda_{t}Q_{t}k_{t+1}+\nu_{t}\left(Q_{t}k_{t+1}-\frac{d_{t+1}}{R_{t}}\right)}{1+\lambda_{t}}+\sigma\mathbb{E}_{t}(\Lambda_{t,t+1}\mu_{t+1})\\
    &=\frac{\lambda_{t}v_{t}+\nu_{t}n_{t}}{1+\lambda_{t}}+\sigma\mathbb{E}_{t}(\Lambda_{t,t+1}\mu_{t+1})\\
    &=\nu_{t}n_{t}+\sigma(1+\lambda_{t})\mathbb{E}_{t}(\Lambda_{t,t+1}\mu_{t+1}),
\end{align*}
where the second equality is by definitions of $n_{t+1}$ and $\mu_{t+1}$, the third equality is by \eqref{eq: banker Euler deposit stationary} and \eqref{eq: banker Euler asset stationary}, the fourth equality is by \eqref{eq: banker CSV value stationary} and the balance sheet constraint, and the fifth equality is by rearrangement of terms. It follows that $\mu_{t}(s^{t})=0$ for all $(t,s^{t})$. Therefore, $v_{t}=\nu_{t}n_{t}$.

\subsubsection{Symmetry of individual decisions}\label{sec: CE symmetry}
The banker's constraints when $n_{t}>0$---the only relevant case---are equivalent to
\begin{gather*}
    Q_{t}\Hat{k}_{t+1}=1+\frac{\Hat{d}_{t+1}}{R_{t}},\qquad
    \Hat{n}_{t+1}=X_{t+1}\Hat{k}_{t+1}-\Hat{d}_{t+1},\qquad
    \Hat{v}_{t}=\mathbb{E}_{t}[\Lambda_{t,t+1}(1-\sigma+\sigma\Hat{v}_{t+1})\Hat{n}_{t+1}],\\
    \Hat{v}_{t}\ge\theta{}Q_{t}\Hat{k}_{t+1},
\end{gather*}
where $\Hat{k}_{t+1}\equiv\frac{k_{t+1}}{n_{t}}$, $\Hat{d}_{t+1}\equiv\frac{d_{t+1}}{n_{t}}$, $\Hat{n}_{t+1}\equiv\frac{n_{t+1}}{n_{t}}$, and $\Hat{v}_{t}\equiv\frac{v_{t}}{n_{t}}$. Observe that none of the normalized constraints depend on $n_{t}$. Furthermore, $n_{t+1}>0$ if and only if $\Hat{n}_{t+1}>0$. Finally $\lambda_{t}(v_{t}-\theta{}Q_{t}k_{t+1})=0$ if and only if $\lambda_{t}(\Hat{v}_{t}-\theta{}Q_{t}\Hat{k}_{t+1})=0$. Consequently, the KKT conditions can be equivalently written in terms of normalized variables and are $n_{t}$-invariant: all Lagrange multipliers are $n_{t}$-invariant, the enforcement constraints bind either for all banks or for none, and old banks are either all solvent at any $s^{t+1}$ or are all insolvent.

\subsection{Proposition \ref{prop: CEA characterization}}\label{sec: proof CEA characterization}
The sequential planning problem is
\begin{equation*}
    \max_{\{C_{t},D_{t+1},K_{t+1},L_{t},V_{t}\}}\mathbb{E}_{0}\left[\sum_{t=0}^\infty\beta^{t}U(C_{t},L_{t})\right]
\end{equation*}
subject to
\begin{align*}
    \widetilde{\nu}_{t}:\quad{}0&=N_{t}-Q(K_{t},K_{t+1},\xi_{t})K_{t+1}+\beta\mathbb{E}_{t}\left(\frac{U_C(C_{t+1},L_{t+1})}{U_C(C_{t},L_{t})}\right)D_{t+1},\\
    \Bar{\gamma}_{t}:\quad{}0&=\mathbb{E}_{t}\left\{\beta\frac{U_C(C_{t+1},L_{t+1})}{U_C(C_{t},L_{t})}[(1-\sigma)(X_{t+1}K_{t+1}-D_{t+1})+\Delta_{t+1}V_{t+1}]\right\}-V_{t},\\
    \widetilde{\lambda}_{t}:\quad{}0&\le{}V_{t}-\theta{}Q(K_{t},K_{t+1},\xi_{t})K_{t+1},\\
    \lambda^L_{t}:\quad{}0&=U_C(C_{t},L_{t})A_{t}F_{L}(\xi_{t}K_{t},L_{t})+U_L(C_{t},L_{t}),\\
    \lambda^Y_{t}:\quad{}0&=A_{t}F(\xi_{t}K_{t},L_{t})-C_{t}-I(K_{t},K_{t+1},\xi_{t}),
\end{align*}
where
\begin{gather*}
    X_{t}\equiv[A_{t}F_{K}(\xi_{t}K_{t},L_{t})+Q(K_{t},K_{t+1},\xi_{t})(1-\delta)]\xi_{t},\\
    N_{t}\equiv\sigma(X_{t}K_{t}-D_{t})+\Bar{N}_0+\Bar{\omega}Q(K_{t},K_{t+1},\xi_{t})K_{t},
\end{gather*}
and $\{\Delta_{t}\}$ is either given or satisfies $\{\Delta_{t}\}=\{\widehat{\sigma}^1_{t}\}$. In the latter case, the partial derivatives are
\begin{gather*}
    \frac{\partial\Delta_{t}}{\partial{}D_{t}}\equiv-\sigma\frac{N^0_{t}}{N_{t}^2}<0,\\
    \frac{\partial\Delta_{t}}{\partial{}K_{t}}\equiv\sigma\frac{N^0_{t}}{N_{t}^2}\{[A_{t}F_{KK,t}\xi_{t}+Q_{1,t}(1-\delta)]\xi_{t}K_{t}+X_{t}\}-\sigma\frac{X_{t}K_{t}-D_{t}}{N_{t}^2}\omega(Q_{1,t}K_{t}+Q_{t}),\\
    \frac{\partial\Delta_{t}}{\partial{}K_{t+1}}\equiv\frac{\sigma{}Q_{2,t}K_{t}[N^0_{t}(1-\delta)\xi_{t}-\Bar{\omega}(X_{t}K_{t}-D_{t})]}{N_{t}^2},\qquad
    \frac{\partial\Delta_{t}}{\partial{}L_{t}}\equiv\sigma\frac{N^0_{t}}{N_{t}^2}A_{t}F_{KL,t}\xi_{t}K_{t}>0,
\end{gather*}
where $N^0_{t}\equiv\Bar{N}_0+\Bar{\omega}Q_{t}K_{t}$. Otherwise, all derivatives are zero.

Define $\nu_{t}\equiv\frac{\widetilde{\nu}_{t}}{U_{C,t}}$, $\gamma_{t}\equiv\frac{\Bar{\gamma}_{t}}{U_{C,t}}$, and $\lambda_{t}\equiv\frac{\widetilde{\lambda}_{t}}{U_{C,t}}$. The FOCs are
\begin{align*}
    C_{t}:\quad{}0&=U_{C,t}+U_{CC,t}[\nu_{t}(N_{t}-Q_{t}K_{t+1})-\gamma_{t}V_{t}]+\lambda^L_{t}(U_{CC,t}A_{t}F_{L,t}+U_{CL,t})-\lambda^Y_{t}\\
    &\quad+\bm{1}_\mathbb{N}(t)U_{CC,t}\{\nu_{t-1}D_{t}+\gamma_{t-1}[(1-\sigma)(X_{t}K_{t}-D_{t})+\Delta_{t}V_{t}]\},\\
    D_{t+1}:\quad\nu_{t}&=\mathbb{E}_{t}\left\{\Lambda_{t,t+1}\left[\left(1-\sigma-\frac{\partial\Delta_{t+1}}{\partial{}D_{t+1}}V_{t+1}\right)\gamma_{t}+\sigma\nu_{t+1}\right]\right\}R_{t},\\
    K_{t+1}:\quad{}0&=\nu_{t}(Q_{2,t}\{[\sigma(1-\delta)\xi_{t}+\omega]K_{t}-K_{t+1}\}-Q_{t})-\lambda_{t}\theta(Q_{2,t}K_{t+1}+Q_{t})\\
    &\quad+\mathbb{E}_{t}(\Lambda_{t,t+1}[(1-\sigma)\gamma_{t}+\sigma\nu_{t+1}]\{[A_{t+1}F_{KK,t+1}\xi_{t+1}+Q_{1,t+1}(1-\delta)]\xi_{t+1}K_{t+1}\\
    &\quad+X_{t+1}\})+\mathbb{E}_{t}\biggl(\Lambda_{t,t+1}\biggl\{\nu_{t+1}[\omega(Q_{1,t+1}K_{t+1}+Q_{t+1})-Q_{1,t+1}K_{t+2}]\\
    &\quad-\lambda_{t+1}\theta{}Q_{1,t+1}K_{t+2}+\lambda^L_{t+1}A_{t+1}F_{KL,t+1}\xi_{t+1}\\
    &\quad+\frac{\lambda^Y_{t+1}}{U_{C,t+1}}(A_{t+1}F_{K,t+1}\xi_{t+1}-I_{1,t+1})\biggr\}\biggr)-\frac{\lambda^Y_{t}}{U_{C,t}}I_{2,t}+\gamma_{t}\mathbb{E}_{t}\left(\Lambda_{t,t+1}\frac{\partial\Delta_{t+1}}{\partial{}K_{t+1}}V_{t+1}\right)\\
    &\quad+\bm{1}_\mathbb{N}(t)\gamma_{t-1}\left[(1-\sigma)Q_{2,t}(1-\delta)\xi_{t}K_{t}+\frac{\partial\Delta_{t}}{\partial{}K_{t+1}}V_{t}\right],\\
    L_{t}:\quad{}0&=U_{L,t}+\nu_{t}[U_{C,t}\sigma{}A_{t}F_{KL,t}\xi_{t}K_{t}+U_{CL,t}(N_{t}-Q_{t}K_{t+1})]-\gamma_{t}U_{CL,t}V_{t}\\
    &\quad+\lambda^L_{t}(U_{CL,t}A_{t}F_{L,t}+U_{C,t}A_{t}F_{LL,t}+U_{LL,t})+\lambda^Y_{t}A_{t}F_{L,t}\\
    &\quad+\bm{1}_\mathbb{N}(t)\biggl[\nu_{t-1}U_{CL,t}D_{t}+\gamma_{t-1}\biggl\{U_{CL,t}[(1-\sigma)(X_{t}K_{t}-D_{t})+\Delta_{t}V_{t}]\\
    &\quad+U_{C,t}\left[(1-\sigma)A_{t}F_{KL,t}\xi_{t}K_{t}+\frac{\partial\Delta_{t}}{\partial{}L_{t}}V_{t}\right]\biggr\}\biggr],\\
    V_{t}:\quad\gamma_{t}&=\bm{1}_\mathbb{N}(t)\Delta_{t}\gamma_{t-1}+\lambda_{t}.
\end{align*}
The complementary slackness conditions are
\begin{equation*}
    0=\lambda_{t}(V_{t}-\theta{}Q_{t}K_{t+1}),\qquad\lambda_{t}\ge0.
\end{equation*}

Conditional on the multipliers $\{\nu_{t},\gamma_{t}\}$, the planner's FOC for $D_{t+1}$ is equivalent to the individual FOC for $d_{t+1}$ if the planner cannot internalize $\Delta$, that is, when $\frac{\partial\Delta_{t+1}}{\partial{}D_{t+1}}=0$. Yet, the original multipliers are generally nonstationary in the CE and stationary in the CEA. For a closer comparison, we must apply a transformation, as we did with the CE. The problem is that, unlike in the CE, we can have $\gamma_{t}=0$. To apply the transformation, we will assume that $\gamma_{t}(s^{t})>0$ for all $(t,s^{t})$. We thus consider $t\ge{}t^*$ such that $\lambda_{t^*}>0$. We will not be relying on this assumption in the computation. Define $\widehat{x}_{t}\equiv\frac{x_{t}}{\gamma_{t}}$ and $\Bar{x}_{t}\equiv\frac{\widehat{x}_{t}}{1-\widehat{\lambda}_{t}}$ for $x\in\{\nu,\lambda,\lambda^L,\lambda^Y\}$. The FOC for $V_{t}$ implies $\frac{\gamma_{t}}{\gamma_{t-1}}=\frac{\bm{1}_\mathbb{N}(t)\Delta_{t}}{1-\widehat{\lambda}_{t}}$, which implies $\frac{\gamma_{t+1}}{\gamma_{t}}=\frac{\Delta_{t+1}}{1-\widehat{\lambda}_{t+1}}$. We can, therefore, write the FOCs for $D_{t+1}$ and $K_{t+1}$ as
\begin{align*}
    \Tilde{\nu}_{t}&=(1+\Tilde{\lambda}_{t})\mathbb{E}_{t}\left[\Lambda_{t,t+1}\left(1-\sigma+\Delta_{t+1}\sigma\Tilde{\nu}_{t+1}-\frac{\partial\Delta_{t+1}}{\partial{}D_{t+1}}V_{t+1}\right)\right]R_{t},\\
    \theta\Tilde{\lambda}_{t}+\Tilde{\nu}_{t}&=(1+\Tilde{\lambda}_{t})\mathbb{E}_{t}[\Lambda_{t,t+1}(1-\sigma+\Delta_{t+1}\sigma\Tilde{\nu}_{t+1})\frac{X_{t+1}}{Q_{t}}]+\Psi^K_{t},
\end{align*}
where
\begin{align*}
    Q_{t}\Psi^K_{t}&\equiv\Tilde{\nu}_{t}Q_{2,t}\{[\sigma(1-\delta)\xi_{t}+\omega]K_{t}-K_{t+1}\}+(1+\Tilde{\lambda}_{t})\mathbb{E}_{t}\left(\Lambda_{t,t+1}\frac{\partial\Delta_{t+1}}{\partial{}K_{t+1}}V_{t+1}\right)-\Tilde{\lambda}_{t}\theta{}Q_{2,t}K_{t+1}\\
    &\quad-\frac{\Tilde{\lambda}^Y_{t}}{U_{C,t}}I_{2,t}+\frac{\bm{1}_\mathbb{N}(t)}{\Delta_{t}}\left[(1-\sigma)Q_{2,t}(1-\delta)\xi_{t}K_{t}+\frac{\partial\Delta_{t}}{\partial{}K_{t+1}}V_{t}\right]\\
    &\quad+(1+\Tilde{\lambda}_{t})\mathbb{E}_{t}\{\Lambda_{t,t+1}(1-\sigma+\Delta_{t+1}\sigma\Tilde{\nu}_{t+1})[A_{t+1}F_{KK,t+1}\xi_{t+1}+Q_{1,t+1}(1-\delta)]\xi_{t+1}K_{t+1}\}\\
    &\quad+(1+\Tilde{\lambda}_{t})\mathbb{E}_{t}\biggl(\Lambda_{t,t+1}\Delta_{t+1}\biggl\{\Tilde{\nu}_{t+1}[\omega(Q_{1,t+1}K_{t+1}+Q_{t+1})-Q_{1,t+1}K_{t+2}]\\
    &\quad-\Tilde{\lambda}_{t+1}\theta{}Q_{1,t+1}K_{t+2}+\Tilde{\lambda}^L_{t+1}A_{t+1}F_{KL,t+1}\xi_{t+1}+\frac{\Tilde{\lambda}^Y_{t+1}}{U_{C,t+1}}(A_{t+1}F_{K,t+1}\xi_{t+1}-I_{1,t+1})\biggr\}\biggr).
\end{align*}

The arguments for bullet points are as follows.
\begin{enumerate}
    \item If the enforcement constraint is always binding at the CEA, the CEA is completely determined by the planner's constraints, and the planner's Euler equations for $D_{t+1}$ and $K_{t+1}$ determine the social Lagrange multipliers $\{\Tilde{\nu}_{t},\Tilde{\lambda}_{t}\}$. The implementability constraints can be formulated in recursive form, constituting a system of functional equations on a state-space with endogenous states $(D,K)$ and exogenous states $(A,\xi)$. If the enforcement constraint is always nonbinding at the CEA, the balance sheet, bank value, and enforcement constraints must be redundant. Indeed, the sequence of history-contingent balance sheet constraints determines $\{D_{t+1}\}$ and the sequence of bank value constraints determines $\{V_{t}\}$ conditional on $\{C_{t},K_{t+1},L_{t}\}$. The remaining labor market clearing and resource constraints are static. Therefore, the CEA must be recursive. In general, however, due to forward-looking constraints, many KKT conditions are different for $t=0$ and $t>0$; hence, the CEA is generally time inconsistent.
    \item If the enforcement constraint is always binding at the CEA, the CEA is determined from the implementability constraints. With $\{\Delta_{t}\}=\{\widehat{\sigma}^1_{t}\}$, those constraints are necessary for the CE. They are also sufficient, since we can use \eqref{eq: bank Euler deposits DE} and \eqref{eq: bank Euler capital DE} to back out the transformed CE multipliers $\{\Bar{\nu_{t}},\Tilde{\lambda}_{t}\}$, where the condition $\mathbb{E}_{t}[\Lambda_{t,t+1}f_{t+1}(\frac{X_{t+1}}{Q_{t}}-R_{t})]\ge0$ applied to $f_{t}\equiv{}1-\sigma+\sigma\Tilde{\nu}_{t}$ guarantees that $\Tilde{\lambda}_{t}\ge0$ for all $(t,s^{t})$. Therefore, the CEA equals the CE allocation, and the latter is constrained efficient.
    \item Given $\{\Delta_{t}\}$, consider any $\Bar{S}^{t}$ that is of positive measure. The definition of the forward-looking bank value---the aggregate bank value constraint---and the fact that $\Delta_{t}(s^{t})<1$ for all $(t,s^{t})$ imply that for all $s^{t}\in\Bar{S}^{t}$, there exists $\{\epsilon_{t+1}(s^{t+1})\}_{s^{t+1}\in{}S^{t+1}\mid{}s^{t}}$, where $\epsilon_{t+1}(s^{t+1})\in(0,1-\Delta_{t+1}(s^{t+1}))$ for all $s^{t+1}\in{}S^{t+1}\mid{}s^{t}$. Let us define $\{\widetilde{\sigma}^1_{t}\}$ as follows. For all $s^{t}\in\Bar{S}^{t}$, set $\widetilde{\sigma}^1_{t+1}(s^{t+1})\equiv\Delta_{t+1}(s^{t+1})+\epsilon_{t+1}(s^{t+1})$ for all $s^{t+1}$ that continue from $s^{t}$. Set $\widetilde{\sigma}^1_{t}(s^{t})\equiv\Delta_{t}(s^{t})$ for all other $(t,s^{t})$. The alternative distribution constructed this way would strictly relax enforcement constraints at $s^{t}\in\Bar{S}^{t}$ and strictly expand the planner's feasible set. Since, conditional on $\{\Delta_{t}\}$, the optimal allocation at $s^{t}\in\Bar{S}^{t}$ is strictly at the boundary of the feasible set, the optimal allocation over the expanded feasible set will be strictly outside the boundary. The strict concavity of $U$ and the positive measure of $\Bar{S}^{t}$ imply that the CEA conditional on $\{\widetilde{\sigma}^1_{t}\}$ must attain strictly greater welfare than the original CEA.
\end{enumerate}

\subsection{Proposition \ref{prop: MCEA characterization}}\label{sec: proof MCEA characterization}
The current planner's best response is
\begin{equation*}
    V^h(S)=\max_{(C,D',K',L)\in\mathcal{G}(S)}{U(C,L)+\beta\mathbb{E}_{z}(\Bar{V}^h(S'))},
\end{equation*}
where $\mathcal{G}:\mathcal{S}\to\mathcal{P}(\mathbb{R}^4_{+})$ is defined by the constraints
\begin{align*}
    \widetilde{\nu}:\quad{}0&=\sigma(X_{S}K-D)+\Bar{N}_0+Q(K,K',\xi)(\Bar{\omega}K-K')+\beta\mathbb{E}_{z}\left(\frac{U_C(\Bar{C}_{S'},\Bar{L}_{S'})}{U_C(C,L)}\right)D',\\
    \widetilde{\lambda}:\quad{}0&\le\beta\mathbb{E}_{z}\left\{\frac{U_C(\Bar{C}_{S'},\Bar{L}_{S'})}{U_C(C,L)}[(1-\sigma)(X_{S'}K'-D')+\Bar{V}^1_{S'}]\right\}-\theta{}Q(K,K',\xi)K',\\
    \lambda^L:\quad{}0&=U_C(C,L)AF_{L}(\xi{}K,L)+U_L(C,L),\\
    \lambda^Y:\quad{}0&=AF(\xi{}K,L)-C-I(K,K',\xi),
\end{align*}
where
\begin{gather*}
    X_{S}\equiv[AF_{K}(\xi{}K,L)+Q(K,K',\xi)(1-\delta)]\xi,\\
    X_{S'}\equiv[A'F_{K}(\xi'K',\Bar{L}_{S'})+Q(K',\Bar{K}'_{S'},\xi')(1-\delta)]\xi'.
\end{gather*}

Define $\nu_{S}\equiv\frac{\widetilde{\nu}_{S}}{U_{C,S}}$ and $\lambda_{S}\equiv\frac{\widetilde{\lambda}_{S}}{U_{C,S}}$. The FOCs are
\begin{align*}
    C:\quad\lambda^Y_{S}&=U_{C,S}+U_{CC,S}[\nu_{S}(N_{S}-Q_{S}K'_{S})-\lambda_{S}\theta{}Q_{S}K'_{S}]+\lambda^L_{S}(U_{CC,S}AF_{L,S}+U_{CL,S}),\\
    D':\quad{}0&=\beta\mathbb{E}_{z}(\Bar{V}^h_{D,S'})+\nu_{S}U_{C,S}\mathbb{E}_{z}(\Lambda_{S,S'})+\lambda_{S}U_{C,S}\mathbb{E}_{z}[\Lambda_{S,S'}(1-\sigma)(-1)]+\Xi^D_{S},\\
    K':\quad{}0&=\beta\mathbb{E}_{z}(\Bar{V}^h_{K,S'})+\nu_{S}U_{C,S}[Q_{2,S}\{[\sigma(1-\delta)\xi+\omega]K-K'_{S}\}-Q_{S}]\\
    &\quad+\lambda_{S}U_{C,S}\{\mathbb{E}_{z}[\Lambda_{S,S'}(1-\sigma)\widetilde{N}_{K,S'}]-\theta(Q_{2,S}K'_{S}+Q_{S})\}-\lambda^Y_{S}I_{2,S}+\Xi^K_{S},\\
    L:\quad{}0&=U_{L,S}+\nu_{S}\left[U_{CL,S}(N_{S}-Q_{S}K'_{S})+U_{C,S}\sigma{}AF_{KL,S}\xi{}K\right]-U_{CL,S}\lambda_{S}\theta{}Q_{S}K'_{S}\\
    &\quad+\lambda^L_{S}[A(U_{CL,S}F_{L,S}+U_{C,S}F_{LL,S})+U_{LL,S}]+\lambda^Y_{S}AF_{L,S},
\end{align*}
where $\widetilde{N}_{K,S}\equiv[AF_{KK,S}\xi+Q_{1,S}(1-\delta)]\xi{}K+X_{S}$, and for $x\in\{D,K\}$,
\begin{multline*}
    \Xi^X_{S}\equiv\beta\nu_{S}\mathbb{E}_{z}(U_{CC,S'}\Bar{C}_{X,S'}+U_{CL,S'}\Bar{L}_{X,S'})D'_{S}\\
    +\beta\lambda_{S}\mathbb{E}_{z}([(U_{CC,S'}\Bar{C}_{X,S'}+U_{CL,S'}\Bar{L}_{X,S'})[(1-\sigma)(X_{S'}K'_{S}-D'_{S})+\Bar{V}^1_{S'}]\\
    +U_{C,S'}\{(1-\sigma)[A'F_{KL,S'}\Bar{L}_{X,S'}+Q_{2,S'}\Bar{K}'_{X,S'}(1-\delta)]\xi'K'_{S}+\Bar{V}^1_{X,S'}\}])
\end{multline*}
is the combined marginal effect of $X'$ on the current planner's Lagrangian through the policy functions of the future planner $\Bar{C}$, $\Bar{L}$, $\Bar{K}'$, and $\Bar{V}^1$. The envelope conditions are
\begin{align*}
    V^h_{D,S}&=-\nu_{S}U_{C,S}\sigma,\\
    V^h_{K,S}&=\nu_{S}U_{C,S}[\sigma\widetilde{N}_{K,S}+\omega(Q_{1,S}K+Q_{S})-Q_{1,S}K'_{S}]-\lambda_{S}U_{C,S}\theta{}Q_{1,S}K'_{S}+\lambda^L_{S}U_{C,S}AF_{KL,S}\xi\\
    &\quad+\lambda^Y_{S}(AF_{K,S}\xi-I_{1,S}).
\end{align*}
The complementary slackness conditions are
\begin{equation*}
    0=\lambda_{S}[\mathbb{E}_{z}\{\Lambda_{S,S'}[(1-\sigma)(X_{S'}K'_{S}-D'_{S})+\Bar{V}^1_{S'}]\}-\theta{}Q_{S}K'_{S}],\qquad\lambda_{S}\ge0.
\end{equation*}
Substituting the envelope conditions in the FOCs for $D'$ and $K'$ and rearranging, we can write
\begin{align*}
    \nu_{S}&=R_{S}\mathbb{E}_{z}\{\Lambda_{S,S'}[(1-\sigma)\lambda_{S}+\sigma\nu_{S'}]\}-\frac{R_{S}\Xi^D_{S}}{U_{C,S}},\\
    \lambda_{S}\theta+\nu_{S}&=\mathbb{E}_{z}\left\{\Lambda_{S,S'}[(1-\sigma)\lambda_{S}+\sigma\nu_{S'}]\frac{X_{S'}}{Q_{S}}\right\}+\Omega^K_{S}+\frac{\Xi^K_{S}}{Q_{S}U_{C,S}},
\end{align*}
where
\begin{multline*}
    Q_{S}\Omega^K_{S}\equiv\nu_{S}Q_{2,S}\{[\sigma(1-\delta)\xi+\omega]K-K'_{S}\}-\lambda_{S}\theta{}Q_{2,S}K'_{S}-\frac{\lambda^Y_{S}}{U_{C,S}}I_{2,S}\\
    +\mathbb{E}_{z}\{\Lambda_{S,S'}[(1-\sigma)\lambda_{S}+\sigma\nu_{S'}][A'F_{KK,S'}\xi'+Q_{1,S'}(1-\delta)]\xi'K'_{S}\}\\
    +\mathbb{E}_{z}\biggl(\Lambda_{S,S'}\biggl\{\nu_{S'}[\omega(Q_{1,S'}K'_{S}+Q_{S'})-Q_{1,S'}K'_{S'}]-\lambda_{S'}\theta{}Q_{1,S'}K'_{S'}+\lambda^L_{S'}A'F_{KL,S'}\xi'\\
    +\frac{\lambda^Y_{S'}}{U_{C,S'}}(A'F_{K,S'}\xi'-I_{1,S'})\biggr\}\biggr).
\end{multline*}

\subsection{Proposition \ref{prop: implementation}}\label{sec: proof implementation}
We will prove the CEA implementation. The arguments for the MCEA implementation are identical.

\subsubsection{Regulated banker's problem}
Consider a regulated CE with affine taxes on bank balance sheet components and a regulatory capital requirement constraint. For a bank $i\in[0,f]$, the modified constraints can be written as
\begin{equation*}
    \widetilde{Q}_{t}k_{i,t+1}=\Bar{n}_{i,t}+\frac{d_{i,t+1}}{\widetilde{R}_{t}},\qquad
    \Bar{n}_{i,t}\ge\kappa_{t}\widetilde{Q}_{t}k_{i,t+1},
\end{equation*}
where
\begin{gather*}
    \widetilde{Q}_{t}\equiv(1+\tau^K_{t})Q_{t},\qquad
    \widetilde{R}_{t}\equiv\frac{R_{t}}{1-\tau^D_{t}},\qquad
    \Bar{n}_{i,t}\equiv\zeta^{j(i)}_{t}+(1+\tau^{j(i)}_{t})n_{i,t},\\
    j(i)=
    \begin{cases}
        1 & i\in[0,\sigma{}f]\\
        0 & i\in(\sigma{}f,f]
    \end{cases}.
\end{gather*}
Hence, the taxes $\tau^D_{t}$ and $\tau^K_{t}$ and capital requirements $\kappa_{t}$ are common to all banks, while the lump-sum transfer $\zeta^{j(i)}_{t}$ and the linear net worth subsidy $\tau^{j(i)}_{t}$ vary between the old and new banks but are equal within each group. To be consistent with \eqref{eq: banking sector balance sheet}, a government budget constraint holds:
\begin{equation}
    \tau^D_{t}\frac{D_{t+1}}{R_{t}}+\tau^K_{t}Q_{t}K_{t+1}=\Tilde{N}_{t}-N_{t},\label{eq: government BC}
\end{equation}
where
\begin{equation}
    \Tilde{N}_{t}\equiv[\sigma\zeta^1_{t}+(1-\sigma)\zeta^0_{t}]f+\tau^1_{t}N^1_{t}+\tau^0_{t}N^0_{t}+N_{t},\label{eq: banker aggregate post-tax net worth}
\end{equation}
and, as before, $N^1_{t}\equiv\sigma(X_{t}K_{t}-D_{t})$ and $N^0_{t}\equiv\Bar{N}_0+\Bar{\omega}Q_{t}K_{t}$. The modified Lagrangian is
\begin{multline*}
    \mathcal{L}=v_{0}+\mathbb{E}_{0}\biggl[\sum_{t=0}^\infty\sigma^{t}\beta^{t}\frac{U_{C,t}}{U_{C,0}}\biggl\{\nu_{t}\left(\Bar{n}_{t}+\frac{d_{t+1}}{\widetilde{R}_{t}}-\widetilde{Q}_{t}k_{t+1}\right)+\gamma_{t}[\mathbb{E}_{t}\{\Lambda_{t,t+1}[(1-\sigma)n_{t+1}+\sigma{}v_{t+1}]\}-v_{t}]\\
    +\lambda_{t}(v_{t}-\theta{}Q_{t}k_{t+1})+\xi_{t}(\Bar{n}_{t}-\kappa_{t}\widetilde{Q}_{t}k_{t+1})\biggr\}\biggr],
\end{multline*}
where $n_{t+1}\equiv{}X_{t+1}k_{t+1}-d_{t+1}$. The FOCs with respect to $v_{t}$, $d_{t+1}$, and $k_{t+1}$ are
\begin{align*}
    \gamma_{t}&=\gamma_{t-1}+\lambda_{t},\qquad\gamma_{-1}=1,\\
    \nu_{t}&=\mathbb{E}_{t}\{\Lambda_{t,t+1}[(1-\sigma)\gamma_{t}+\sigma(\nu_{t+1}+\xi_{t+1})(1+\tau^1_{t+1})]\}\widetilde{R}_{t},\\
    \theta\lambda_{t}+\frac{\widetilde{Q}_{t}}{Q_{t}}(\kappa_{t}\xi_{t}+\nu_{t})&=\mathbb{E}_{t}\{\Lambda_{t,t+1}[(1-\sigma)\gamma_{t}+\sigma(\nu_{t+1}+\xi_{t+1})(1+\tau^1_{t+1})]\frac{X_{t+1}}{Q_{t}}\}.
\end{align*}
The complementary slackness conditions are
\begin{align*}
    0&=\lambda_{t}(v_{t}-\theta{}Q_{t}k_{t+1}),\qquad\lambda_{t}\ge0,\qquad{}v_{t}\ge\theta{}Q_{t}k_{t+1},\\
    0&=\xi_{t}(\Bar{n}_{t}-\kappa_{t}\widetilde{Q}_{t}k_{t+1}),\qquad\xi_{t}\ge0,\qquad\Bar{n}_{t}\ge\kappa_{t}\widetilde{Q}_{t}k_{t+1}.
\end{align*}
Note that any bank $i\in[0,f]$ surviving from $t\ge0$ has $j(i)=1$ at $t+1$. Using the transformation $\Tilde{\nu}_{t}\equiv\frac{\nu_{t}}{\gamma_{t-1}}$, $\Tilde{\lambda}_{t}\equiv\frac{\lambda_{t}}{\gamma_{t-1}}$, $\Tilde{\xi}_{t}\equiv\frac{\xi_{t}}{\gamma_{t-1}}$, and noting that $\frac{\gamma_{t}}{\gamma_{t-1}}\equiv{}1+\Tilde{\lambda}_{t}$, we get the stationary KKT conditions
\begin{align}
    \Tilde{\nu}_{t}&=(1+\Tilde{\lambda}_{t})\mathbb{E}_{t}\{\Lambda_{t,t+1}[1-\sigma+\sigma(\Tilde{\nu}_{t+1}+\Tilde{\xi}_{t+1})(1+\tau^1_{t+1})]\}\widetilde{R}_{t},\label{eq: banker Euler deposit regulated}\\
    \theta\Tilde{\lambda}_{t}+\frac{\widetilde{Q}_{t}}{Q_{t}}(\kappa_{t}\Tilde{\xi}_{t}+\Tilde{\nu}_{t})&=(1+\Tilde{\lambda}_{t})\mathbb{E}_{t}\{\Lambda_{t,t+1}[1-\sigma+\sigma(\Tilde{\nu}_{t+1}+\Tilde{\xi}_{t+1})(1+\tau^1_{t+1})]\frac{X_{t+1}}{Q_{t}}\},\label{eq: banker Euler asset regulated}\\
    0&=\Tilde{\lambda}_{t}(v_{t}-\theta{}Q_{t}k_{t+1}),\qquad\Tilde{\lambda}_{t}\ge0,\qquad{}v_{t}\ge\theta{}Q_{t}k_{t+1},\nonumber\\
    0&=\Tilde{\xi}_{t}(\Bar{n}_{t}-\kappa_{t}\widetilde{Q}_{t}k_{t+1}),\qquad\Tilde{\xi}_{t}\ge0,\qquad\Bar{n}_{t}\ge\kappa_{t}\widetilde{Q}_{t}k_{t+1}.\nonumber
\end{align}
Define $\mu_{i,t}\equiv{}v_{i,t}-(\Tilde{\nu}_{t}+\Tilde{\xi}_{t})(1+\tau^{j(i)}_{t})n_{i,t}$. Using the same approach as in Section \ref{sec: CE value function}, we find
\begin{equation*}
    \mu_{i,t}=(\Tilde{\nu}_{t}+\Tilde{\xi}_{t})\zeta^{j(i)}_{t}+\sigma(1+\Tilde{\lambda}_{t})\mathbb{E}_{t}(\Lambda_{t,t+1}\mu_{i,t+1}).
\end{equation*}

\subsubsection{Symmetry of individual decisions}
We will focus on the CEA implementation when the individual bankers' optimal decisions with regulation remain symmetric. Consequently, the Lagrange multipliers are common across all banks, and the value function intercept $\mu_{i,t}$ varies only between the groups of old and new banks, satisfying
\begin{align}
    \mu^1_{t}&=(\Tilde{\nu}_{t}+\Tilde{\xi}_{t})\zeta^1_{t}+\sigma(1+\Tilde{\lambda}_{t})\mathbb{E}_{t}(\Lambda_{t,t+1}\mu^1_{t+1}),\label{eq: value function intercept old}\\
    \mu^0_{t}&=(\Tilde{\nu}_{t}+\Tilde{\xi}_{t})\zeta^0_{t}+\sigma(1+\Tilde{\lambda}_{t})\mathbb{E}_{t}(\Lambda_{t,t+1}\mu^1_{t+1}).\label{eq: value function intercept new}
\end{align}
A sufficient condition for symmetry is linear taxation of surviving banks. If $\zeta^1_{t}=0$ for all $(t,s^{t})$, all constraints can be written in terms of the ratios $\Hat{d}_{t+1}\equiv\frac{d_{t+1}}{\Bar{n}_{t}}$, $\Hat{n}_{t+1}\equiv\frac{n_{t+1}}{\Bar{n}_{t}}$, $\Hat{k}_{t+1}\equiv\frac{k_{t+1}}{\Bar{n}_{t}}$, and $\Hat{v}_{t}\equiv\frac{v_{t}}{\Bar{n}_{t}}$:
\begin{gather*}
    \widetilde{Q}_{t}\Hat{k}_{t+1}=1+\frac{\Hat{d}_{t+1}}{\widetilde{R}_{t}},\qquad
    \Hat{n}_{t+1}=X_{t+1}\Hat{k}_{t+1}-\Hat{d}_{t+1},\\
    \Hat{v}_{t}=\mathbb{E}_{t}\{\Lambda_{t,t+1}[1-\sigma+\sigma\Hat{v}_{t+1}(1+\tau^1_{t+1})]\Hat{n}_{t+1}\},\qquad
    \Hat{v}_{t}\ge\theta{}Q_{t}\Hat{k}_{t+1},\qquad
    1\ge\kappa_{t}\widetilde{Q}_{t}\Hat{k}_{t+1}.
\end{gather*}
The constraint set is independent of $\Bar{n}_{t}$; therefore, as in Section \ref{sec: CE symmetry}, the enforcement constraint regimes, the solvency regimes, and the Lagrange multipliers $\Tilde{\nu}_{t}$, $\Tilde{\lambda}_{t}$, and $\Tilde{\xi}_{t}$ are independent of $\Bar{n}_{t}$. Using the same arguments as in Section \ref{sec: CE value function}, we conclude that $\mu^1_{t}=0$ for all $(t,s^{t})$, which implies $\mu^0_{t}=(\Tilde{\nu}_{t}+\Tilde{\xi}_{t})\zeta^0_{t}$.

\subsubsection{Implementability constraints}\label{sec: implementability constraints}
Conditional on the CEA, the policy that implements it $\{\zeta^0_{t},\zeta^1_{t},\kappa_{t},\tau^D_{t},\tau^K_{t},\tau^0_{t},\tau^1_{t}\}$ and the scaled regulated CE multipliers $\{\Tilde{\lambda}_{t},\Tilde{\nu}_{t},\Tilde{\xi}_{t}\}$ must satisfy the government budget constraint \eqref{eq: government BC}, the stationary regulated Euler equations \eqref{eq: banker Euler deposit regulated} and \eqref{eq: banker Euler asset regulated}, the aggregate complementary slackness conditions
\begin{align}
    0&=\Tilde{\lambda}_{t}(V_{t}-\theta{}Q_{t}K_{t+1}),\qquad\Tilde{\lambda}_{t}\ge0,\label{eq: banker CSC value regulated}\\
    0&=\Tilde{\xi}_{t}(\Tilde{N}_{t}-\kappa_{t}\widetilde{Q}_{t}K_{t+1}),\qquad\Tilde{\xi}_{t}\ge0,\qquad\Tilde{N}_{t}\ge\kappa_{t}\widetilde{Q}_{t}K_{t+1},\label{eq: banker CSC equity regulated}
\end{align}
and the distribution consistency condition
\begin{equation}\label{eq: distribution consistency}
    V^1_{t}=\Delta_{t}V_{t},
\end{equation}
where
\begin{align}
    V^1_{t}&=\sigma{}f\mu^1_{t}+(\Tilde{\nu}_{t}+\Tilde{\xi}_{t})(1+\tau^1_{t})N^1_{t},\label{eq: aggregate old bank value}\\
    V_{t}&=[\sigma\mu^1_{t}+(1-\sigma)\mu^0_{t}]f+(\Tilde{\nu}_{t}+\Tilde{\xi}_{t})(\tau^1_{t}N^1_{t}+\tau^0_{t}N^0_{t}+N_{t}).\label{eq: aggregate bank value regulated}
\end{align}
Note that \eqref{eq: aggregate EC} is guaranteed by the CEA, being redundant in \eqref{eq: banker CSC value regulated}.

Observe that \eqref{eq: banker Euler deposit regulated} must be equivalent to \eqref{eq: aggregate bank value regulated} conditional on other equations and the CEA. Indeed, we have already shown that the former implies the latter when deriving the value function. Moreover, since the CEA satisfies \eqref{eq: aggregate bank value},
\begin{align*}
    V_{t}&=\mathbb{E}_{t}\{\Lambda_{t,t+1}[(1-\sigma)(X_{t+1}K_{t+1}-D_{t+1})+V^1_{t+1}]\}\\
    &=\mathbb{E}_{t}\{\Lambda_{t,t+1}[1-\sigma+\sigma(\Tilde{\nu}_{t+1}+\Tilde{\xi}_{t+1})(1+\tau^1_{t+1})](X_{t+1}K_{t+1}-D_{t+1})\}+f\sigma\mathbb{E}_{t}(\Lambda_{t,t+1}\mu^1_{t+1})\\
    &=(\Tilde{\nu}_{t}+\Tilde{\xi}_{t})\Tilde{N}_{t}+\Tilde{\nu}_{t}\frac{D_{t+1}}{\widetilde{R}_{t}}-(1+\Tilde{\lambda}_{t})\mathbb{E}_{t}\{\Lambda_{t,t+1}[1-\sigma+\sigma(\Tilde{\nu}_{t+1}+\Tilde{\xi}_{t+1})(1+\tau^1_{t+1})]\}D_{t+1}\\
    &\quad+[\sigma\mu^1_{t}+(1-\sigma)\mu^0_{t}]f-(\Tilde{\nu}_{t}+\Tilde{\xi}_{t})[\sigma\zeta^1_{t}+(1-\sigma)\zeta^0_{t}]f\\
    &=V_{t}+\Tilde{\nu}_{t}\frac{D_{t+1}}{\widetilde{R}_{t}}-(1+\Tilde{\lambda}_{t})\mathbb{E}_{t}\{\Lambda_{t,t+1}[1-\sigma+\sigma(\Tilde{\nu}_{t+1}+\Tilde{\xi}_{t+1})(1+\tau^1_{t+1})]\}D_{t+1},
\end{align*}
where the second equality uses \eqref{eq: aggregate old bank value} and the definition of $N^1_{t}$, the third equality uses \eqref{eq: banking sector balance sheet}, \eqref{eq: government BC}, and \eqref{eq: banker Euler asset regulated}--\eqref{eq: banker CSC equity regulated}, and the fourth equality uses \eqref{eq: banker aggregate post-tax net worth} and \eqref{eq: aggregate bank value regulated}. Hence, \eqref{eq: banker Euler deposit regulated} is satisfied. This is convenient, since \eqref{eq: aggregate bank value regulated} allows expressing $\Tilde{\nu}_{t}$ explicitly, while \eqref{eq: banker Euler deposit regulated} does not.

The six equations \eqref{eq: government BC}, \eqref{eq: banker Euler deposit regulated} or \eqref{eq: aggregate bank value regulated}, \eqref{eq: banker Euler asset regulated}, and \eqref{eq: banker CSC value regulated}--\eqref{eq: distribution consistency} determine ten unknowns, so there are as many as four degrees of freedom and many alternative implementation mechanisms.

\subsubsection{Implementation under linear taxation of surviving banks}
To guarantee the symmetry of individual decisions, we set $\zeta^1_{t}=0$ for all $(t,s^{t})$, which implies $\mu^1_{t}=0$ and $\mu^0_{t}=(\Tilde{\nu}_{t}+\Tilde{\xi}_{t})\zeta^0_{t}$. In this case, the implementability conditions \eqref{eq: government BC} with \eqref{eq: banker aggregate post-tax net worth}, \eqref{eq: banker Euler asset regulated}, \eqref{eq: banker CSC value regulated}--\eqref{eq: distribution consistency}, and \eqref{eq: aggregate bank value regulated} become
\begin{align*}
    \tau^D_{t}\frac{D_{t+1}}{R_{t}}+\tau^K_{t}Q_{t}K_{t+1}&=\Tilde{N}_{t}-N_{t},\\
    \Tilde{N}_{t}&=(1-\sigma)f\zeta^0_{t}+\tau^1_{t}N^1_{t}+\tau^0_{t}N^0_{t}+N_{t},\\
    \theta\Tilde{\lambda}_{t}+(1+\tau^K_{t})(\kappa_{t}\Tilde{\xi}_{t}+\Tilde{\nu}_{t})&=(1+\Tilde{\lambda}_{t})\mathbb{E}_{t}\{\Lambda_{t,t+1}[1-\sigma+\sigma(\Tilde{\nu}_{t+1}+\Tilde{\xi}_{t+1})(1+\tau^1_{t+1})]\frac{X_{t+1}}{Q_{t}}\},\\
    0&=\Tilde{\lambda}_{t}(V_{t}-\theta{}Q_{t}K_{t+1}),\qquad\Tilde{\lambda}_{t}\ge0,\\
    0&=\Tilde{\xi}_{t}[\Tilde{N}_{t}-\kappa_{t}(1+\tau^K_{t})Q_{t}K_{t+1}],\qquad\Tilde{\xi}_{t}\ge0,\\
    \Tilde{N}_{t}&\ge\kappa_{t}(1+\tau^K_{t})Q_{t}K_{t+1},\\
    0&=(1+\tau^1_{t})N^1_{t}-\Delta_{t}\Tilde{N}_{t},\\
    V_{t}&=(\Tilde{\nu}_{t}+\Tilde{\xi}_{t})\Tilde{N}_{t}.
\end{align*}

First, we express $\Tilde{N}_{t}$ from \eqref{eq: government BC} and use \eqref{eq: distribution consistency}, \eqref{eq: banker aggregate post-tax net worth}, and \eqref{eq: aggregate bank value regulated} to solve for $\tau^1_{t}$, $\zeta^0_{t}$ or $\tau^0_{t}$, and $\Tilde{\nu}_{t}$ in terms of $(\tau^D_{t},\tau^K_{t},\Tilde{\xi}_{t})$ as follows:
\begin{gather*}
    \Tilde{N}_{t}=\tau^D_{t}\frac{D_{t+1}}{R_{t}}+\tau^K_{t}Q_{t}K_{t+1}+N_{t},\qquad
    \tau^1_{t}=\frac{\Delta_{t}\Tilde{N}_{t}}{N^1_{t}}-1,\\
    \frac{(1-\sigma)f}{N^0_{t}}\zeta^0_{t}+\tau^0_{t}=\frac{\Tilde{N}_{t}-N_{t}-\tau^1_{t}N^1_{t}}{N^0_{t}},\qquad
    \Tilde{\nu}_{t}=\frac{V_{t}}{\Tilde{N}_{t}}-\Tilde{\xi}_{t}.
\end{gather*}
Observe that $\{\zeta^0_{t}\}$ and $\{\tau^0_{t}\}$ are interchangeable. For example, we can set $\zeta^0_{t}=0$ for all $(t,s^{t})$ w.l.o.g.

Define the unregulated marginal cost of deposits and the marginal benefit of assets (capital) 
\begin{align}
    M^D_{t}&\equiv\mathbb{E}_{t}\left[\Lambda_{t,t+1}\left(1-\sigma+\sigma\frac{V^1_{t+1}}{N^1_{t+1}}\right)\right]R_{t},\label{eq: marginal cost of deposits}\\
    M^K_{t}&\equiv\mathbb{E}_{t}\left[\Lambda_{t,t+1}\left(1-\sigma+\sigma\frac{V^1_{t+1}}{N^1_{t+1}}\right)\frac{X_{t+1}}{Q_{t}}\right],\label{eq: marginal benefit of capital}
\end{align}
which are proportional to the conditional expectations in \eqref{eq: banker Euler deposit regulated} and \eqref{eq: banker Euler asset regulated} given \eqref{eq: aggregate old bank value}. Note that $M^D_{t}$ and $M^K_{t}$ are entirely the functions of allocations and can be evaluated at the CEA. Moreover, by definition,
\begin{align*}
    V_{t}&=\mathbb{E}_{t}\{\Lambda_{t,t+1}[(1-\sigma)(X_{t+1}K_{t+1}-D_{t+1})+V^1_{t+1}]\}\\
    &=\mathbb{E}_{t}\left[\Lambda_{t,t+1}\left(1-\sigma+\sigma\frac{V^1_{t+1}}{N^1_{t+1}}\right)(X_{t+1}K_{t+1}-D_{t+1})\right]\\
    &=M^K_{t}Q_{t}K_{t+1}-M^D_{t}\frac{D_{t+1}}{R_{t}},
\end{align*}
which implies,
\begin{equation*}
    M^K_{t}=\frac{V_{t}}{Q_{t}K_{t+1}}+\frac{M^D_{t}}{Q_{t}K_{t+1}}\frac{D_{t+1}}{R_{t}}>\theta.
\end{equation*}

We are left with three implementability conditions \eqref{eq: banker Euler asset regulated}, \eqref{eq: banker CSC value regulated}, and \eqref{eq: banker CSC equity regulated} putting restrictions on five unknowns $\{\kappa_{t},\tau^D_{t},\tau^K_{t},\Tilde{\lambda}_{t},\Tilde{\xi}_{t}\}$. Consider several restricted sets of instruments that allow to pin down the policy uniquely.

\paragraph{Implementation with a minimal set of taxes}
Conditional on $(\tau^D_{t},\tau^K_{t})$, we can set $\kappa_{t}<\frac{\Tilde{N}_{t}}{(1+\tau^K_{t})Q_{t}K_{t+1}}$, so that the capital requirement constraint is not binding and $\Tilde{\xi}_{t}=0$. We need to determine $\{\tau^D_{t},\tau^K_{t},\Tilde{\lambda}_{t}\}$ that satisfy
\begin{align*}
    \theta\Tilde{\lambda}_{t}+\frac{(1+\tau^K_{t})V_{t}}{\tau^D_{t}\frac{D_{t+1}}{R_{t}}+\tau^K_{t}Q_{t}K_{t+1}+N_{t}}&=(1+\Tilde{\lambda}_{t})M^K_{t},\\
    0&=\Tilde{\lambda}_{t}(V_{t}-\theta{}Q_{t}K_{t+1}),\qquad\Tilde{\lambda}_{t}\ge0.
\end{align*}
If $V_{t}>\theta{}Q_{t}K_{t+1}$, $(\tau^D_{t},\tau^K_{t})$ must be such that $\Tilde{\lambda}_{t}=0$. If $V_{t}=\theta{}Q_{t}K_{t+1}$, $(\tau^D_{t},\tau^K_{t})$ must be such that $\Tilde{\lambda}_{t}\ge0$. Hence, $(\tau^D_{t},\tau^K_{t})$ must satisfy
\begin{equation*}
    \tau^D_{t}M^K_{t}+\tau^K_{t}M^D_{t}\le{}M^K_{t}-M^D_{t},\qquad\text{with equality if }V_{t}>\theta{}Q_{t}K_{t+1}.
\end{equation*}
Clearly, $\{\tau^D_{t}\}$ and $\{\tau^K_{t}\}$ are interchangeable, and the two simplest options are to set
\begin{equation*}
    \tau^K_{t}=0,\qquad\tau^D_{t}\le1-\frac{M^D_{t}}{M^K_{t}},\qquad\text{with equality if }V_{t}>\theta{}Q_{t}K_{t+1},
\end{equation*}
or
\begin{equation*}
    \tau^D_{t}=0,\qquad\tau^K_{t}\le\frac{M^K_{t}}{M^D_{t}}-1,\qquad\text{with equality if }V_{t}>\theta{}Q_{t}K_{t+1}.
\end{equation*}

\paragraph{Implementation with capital requirements and a minimal set of taxes}
We need to determine $\{\kappa_{t},\tau^D_{t},\tau^K_{t},\Tilde{\lambda}_{t},\Tilde{\xi}_{t}\}$ that satisfy
\begin{align*}
    \theta\Tilde{\lambda}_{t}+(1+\tau^K_{t})\left[\frac{V_{t}}{\Tilde{N}_{t}}-\Tilde{\xi}_{t}(1-\kappa_{t})\right]&=(1+\Tilde{\lambda}_{t})M^K_{t},\\
    0&=\Tilde{\lambda}_{t}(V_{t}-\theta{}Q_{t}K_{t+1}),\qquad\Tilde{\lambda}_{t}\ge0,\\
    0&=\Tilde{\xi}_{t}[\Tilde{N}_{t}-\kappa_{t}(1+\tau^K_{t})Q_{t}K_{t+1}],\qquad\Tilde{\xi}_{t}\ge0,\\
    &\quad\Tilde{N}_{t}\ge\kappa_{t}(1+\tau^K_{t})Q_{t}K_{t+1}.
\end{align*}
Observe that \eqref{eq: banker Euler asset regulated} is the only way to pin down both $\Tilde{\lambda}_{t}$ and $\Tilde{\xi}_{t}$. Hence, the enforcement and capital requirement constraints cannot be binding at the same time. Moreover, we generally need to use $\tau^D_{t}$ or $\tau^K_{t}$ to guarantee that $\Tilde{\lambda}_{t}\ge0$ or $\Tilde{\xi}_{t}\ge0$, although the latter inequalities may hold with a zero tax. Furthermore, the enforcement and capital requirement constraints should not have slack at the same time, for our intention here is to use the capital requirement constraint instead of taxation to match the wedge between the CEA and the CE.

If $V_{t}>\theta{}Q_{t}K_{t+1}$, we must have $\Tilde{\lambda}_{t}=0$, the capital requirement constraint must be binding, and $\Tilde{\xi_{t}}$ is determined from \eqref{eq: banker Euler asset regulated}. We need the same restriction on $(\tau^D_{t},\tau^K_{t})$ to guarantee $\Tilde{\xi_{t}}\ge0$ as we did before to guarantee $\Tilde{\lambda}_{t}\ge0$. Formally,
\begin{equation*}
    \kappa_{t}=\frac{\Tilde{N}_{t}}{(1+\tau^K_{t})Q_{t}K_{t+1}},\qquad
    \tau^D_{t}M^K_{t}+\tau^K_{t}M^D_{t}\le{}M^K_{t}-M^D_{t}.
\end{equation*}
If $V_{t}=\theta{}Q_{t}K_{t+1}$, the regulatory constraint must have slack, so we set the taxes as in the previous implementation strategy and set
\begin{equation*}
    \kappa_{t}<\frac{\Tilde{N}_{t}}{(1+\tau^K_{t})Q_{t}K_{t+1}}.
\end{equation*}

\subsubsection{Optimal policy}
The Ramsey problem is to choose the allocation $\{C_{t},D_{t+1},K_{t+1},L_{t},V_{t}\}$, policy $\{\zeta^0_{t},\kappa_{t},\tau^D_{t},\\\tau^K_{t},\tau^0_{t},\tau^1_{t}\}$, and scaled Lagrange multipliers $\{\Tilde{\lambda}_{t},\Tilde{\nu}_{t},\Tilde{\xi}_{t}\}$ that solve
\begin{equation*}
    \max\mathbb{E}_{0}\left[\sum_{t=0}^\infty\beta^{t}U(C_{t},L_{t})\right]
\end{equation*}
subject to \eqref{eq: household labor supply}, \eqref{eq: banking sector balance sheet}, \eqref{eq: market clearing final good}, \eqref{eq: aggregate EC}, \eqref{eq: government BC} with \eqref{eq: banker aggregate post-tax net worth}, \eqref{eq: banker Euler deposit regulated}, \eqref{eq: banker Euler asset regulated}, \eqref{eq: banker CSC value regulated}, \eqref{eq: banker CSC equity regulated}, \eqref{eq: aggregate old bank value}, \eqref{eq: aggregate bank value regulated}, the definitions of several auxiliary variables, $\zeta^1_{t}=0$ (hence, $\mu^1_{t}=0$ and $\mu^0_{t}=(\Tilde{\nu}_{t}+\Tilde{\xi}_{t})\zeta^0_{t}$) for all $(t,s^{t})$, and using \eqref{eq: household Euler deposits}, \eqref{eq: banking sector net worth}--\eqref{eq: market clearing capital}, and \eqref{eq: distribution consistency} to solve for $R_{t}$, $N_{t}$, $W_{t}$, $X_{t}$, $Q_{t}$, $I_{t}$, and $V^1_{t}$.

First, note that conditional on other constraints, \eqref{eq: aggregate bank value regulated} is equivalent to \eqref{eq: aggregate bank value}. Section \ref{sec: implementability constraints} showed that \eqref{eq: banker Euler deposit regulated} is equivalent to \eqref{eq: aggregate bank value regulated}, and it follows from an identical derivation that \eqref{eq: aggregate bank value} implies \eqref{eq: aggregate bank value regulated} when \eqref{eq: banker Euler deposit regulated} holds. Conversely, multiply both sides of \eqref{eq: banker Euler deposit regulated} by $\frac{D_{t+1}}{\widetilde{R}_{t}}$ and both sides of \eqref{eq: banker Euler asset regulated} by $Q_{t}K_{t+1}$, subtract the former from the latter, and use \eqref{eq: banking sector balance sheet}, \eqref{eq: government BC}, \eqref{eq: banker CSC value regulated}, \eqref{eq: banker CSC equity regulated}, \eqref{eq: aggregate old bank value}, and \eqref{eq: aggregate bank value regulated} to get \eqref{eq: aggregate bank value}.

Now we show that \eqref{eq: government BC} with \eqref{eq: banker aggregate post-tax net worth}, \eqref{eq: banker Euler deposit regulated}, \eqref{eq: banker Euler asset regulated}, \eqref{eq: banker CSC value regulated}, \eqref{eq: banker CSC equity regulated}, and \eqref{eq: aggregate old bank value} can be used to construct $\{\zeta^0_{t},\kappa_{t},\tau^D_{t},\tau^K_{t},\tau^0_{t},\tau^1_{t}\}$ and $\{\Tilde{\lambda}_{t},\Tilde{\nu}_{t},\Tilde{\xi}_{t}\}$ conditional on $\{C_{t},D_{t+1},K_{t+1},L_{t},V_{t}\}$. Indeed, express $\Tilde{N}_{t}$ from \eqref{eq: government BC} and use \eqref{eq: banker Euler deposit regulated} with the definition \eqref{eq: marginal cost of deposits}, \eqref{eq: aggregate old bank value}, and \eqref{eq: banker aggregate post-tax net worth} to solve for $\Tilde{\nu}_{t}$, $\tau^1_{t}$, and $\zeta^0_{t}$ or $\tau^0_{t}$ conditional on the allocation and $(\tau^D_{t},\tau^K_{t},\Tilde{\lambda}_{t},\Tilde{\xi}_{t})$:
\begin{gather*}
    \Tilde{N}_{t}=\tau^D_{t}\frac{D_{t+1}}{R_{t}}+\tau^K_{t}Q_{t}K_{t+1}+N_{t},\qquad
    \Tilde{\nu}_{t}=\frac{1+\Tilde{\lambda}_{t}}{1-\tau^D_{t}}M^D_{t},\qquad
    \tau^1_{t}=\frac{V^1_{t}}{N^1_{t}(\Tilde{\nu}_{t}+\Tilde{\xi}_{t})}-1,\\
    \frac{(1-\sigma)f}{N^0_{t}}\zeta^0_{t}+\tau^0_{t}=\frac{\Tilde{N}_{t}-\tau^1_{t}N^1_{t}-N_{t}}{N^0_{t}}.
\end{gather*}
Given this conditional solution, we are left to construct $\{\kappa_{t},\tau^D_{t},\tau^K_{t}\}$ and $\{\Tilde{\lambda}_{t},\Tilde{\xi}_{t}\}$ from \eqref{eq: banker Euler asset regulated}, \eqref{eq: banker CSC value regulated}, and \eqref{eq: banker CSC equity regulated}:
\begin{align*}
    \theta\Tilde{\lambda}_{t}+\frac{\widetilde{Q}_{t}}{Q_{t}}(\kappa_{t}\Tilde{\xi}_{t}+\Tilde{\nu}_{t})&=(1+\Tilde{\lambda}_{t})M^K_{t},\\
    0&=\Tilde{\lambda}_{t}(V_{t}-\theta{}Q_{t}K_{t+1}),\qquad\Tilde{\lambda}_{t}\ge0,\\
    0&=\Tilde{\xi}_{t}(\Tilde{N}_{t}-\kappa_{t}\widetilde{Q}_{t}K_{t+1}),\qquad\Tilde{\xi}_{t}\ge0,\qquad\Tilde{N}_{t}\ge\kappa_{t}\widetilde{Q}_{t}K_{t+1},
\end{align*}
where we used the definition \eqref{eq: marginal benefit of capital}.
\begin{enumerate}
    \item If $V_{t}>\theta{}Q_{t}K_{t+1}$, set $\Tilde{\lambda}_{t}=0$ and do either of the following:
    \begin{enumerate}
        \item Conditional $(\tau^D_{t},\tau^K_{t})$, choose $\kappa_{t}<\frac{\Tilde{N}_{t}}{\widetilde{Q}_{t}K_{t+1}}$ and set $\Tilde{\xi}_{t}=0$. Choose any $(\tau^D_{t},\tau^K_{t})$ that satisfy \eqref{eq: banker Euler asset regulated}.
        \item Conditional $(\tau^D_{t},\tau^K_{t})$, set $\kappa_{t}=\frac{\Tilde{N}_{t}}{\widetilde{Q}_{t}K_{t+1}}$ and solve for $\Tilde{\xi}_{t}$ from \eqref{eq: banker Euler asset regulated}. Choose any $(\tau^D_{t},\tau^K_{t})$ that imply $\Tilde{\xi}_{t}\ge0$.
    \end{enumerate}
    \item If $V_{t}=\theta{}Q_{t}K_{t+1}$, conditional on $(\tau^D_{t},\tau^K_{t})$, set $\kappa_{t}<\frac{\Tilde{N}_{t}}{\widetilde{Q}_{t}K_{t+1}}$ and $\Tilde{\xi}_{t}=0$ and solve for $\Tilde{\lambda}_{t}$ from \eqref{eq: banker Euler asset regulated}. Choose any $(\tau^D_{t},\tau^K_{t})$ that imply $\Tilde{\lambda}_{t}\ge0$.
\end{enumerate}

Therefore, the Ramsey problem reduces to choosing the allocation $\{C_{t},D_{t+1},K_{t+1},L_{t},V_{t}\}$, and \eqref{eq: government BC} with \eqref{eq: banker aggregate post-tax net worth}, \eqref{eq: banker Euler deposit regulated}, \eqref{eq: banker Euler asset regulated}, \eqref{eq: banker CSC value regulated}, \eqref{eq: banker CSC equity regulated}, and \eqref{eq: aggregate old bank value} are redundant. The resulting problem is identical to the CEA planning problem in Proposition \ref{prop: CEA characterization}.

\paragraph{Optimal time-consistent policy}
The argument is similar to that for the Ramsey problem, but the planner takes as given the policy functions of the future planner. The planner's objective is
\begin{equation*}
    V^h(S)=\max_{(C,D',K',L,V,\zeta^0,\kappa,\tau^D,\tau^K,\tau^0,\tau^1,\Tilde{\lambda},\Tilde{\nu},\Tilde{\xi})\in\mathcal{G}(S)}{U(C,L)+\beta\mathbb{E}_{z}(\Bar{V}^h(S'))},
\end{equation*}
where $S=(D,K,z)$, $z=(A,\xi)$, and the correspondence $\mathcal{G}$ is defined by the same set of constraints as in the Ramsey problem, accounting for future policy functions that are taken as given. As in the Ramsey problem, conditional on other constraints, \eqref{eq: aggregate bank value regulated} is equivalent to \eqref{eq: aggregate bank value}. We need to show that \eqref{eq: government BC} with \eqref{eq: banker aggregate post-tax net worth}, \eqref{eq: banker Euler deposit regulated}, \eqref{eq: banker Euler asset regulated}, \eqref{eq: banker CSC value regulated}, \eqref{eq: banker CSC equity regulated}, and \eqref{eq: aggregate old bank value} are redundant. The only dynamic constraints here are \eqref{eq: banker Euler deposit regulated} and \eqref{eq: banker Euler asset regulated}, which can be written in recursive form, taking as given the future policy functions, as
\begin{align*}
    (1-\tau^D)\Tilde{\nu}&=(1+\Tilde{\lambda})\mathbb{E}_{z}\left\{\beta\frac{U_C(\Bar{C}_{S'},\Bar{L}_{S'})}{U_C(C,L)}[1-\sigma+\sigma(\Bar{\Tilde{\nu}}_{S'}+\Bar{\Tilde{\xi}}_{S'})(1+\Bar{\tau}^1_{S'})]\right\}R_{S},\\
    \theta\Tilde{\lambda}+(1+\tau^K)(\kappa\Tilde{\xi}+\Tilde{\nu})&=(1+\Tilde{\lambda})\mathbb{E}_{z}\left\{\beta\frac{U_C(\Bar{C}_{S'},\Bar{L}_{S'})}{U_C(C,L)}[1-\sigma+\sigma(\Bar{\Tilde{\nu}}_{S'}+\Bar{\Tilde{\xi}}_{S'})(1+\Bar{\tau}^1_{S'})]\frac{X_{S'}}{Q_{S}}\right\},
\end{align*}
where
\begin{gather*}
    R_{S}\equiv\frac{U_C(C,L)}{\beta\mathbb{E}_{z}U_C(\Bar{C}_{S'},\Bar{L}_{S'})},\qquad
    Q_{S}\equiv\left(\Phi'\left\{\Phi^{-1}\left[\frac{K'}{K}-(1-\delta)\xi\right]\right\}\right)^{-1},\\
    X_{S'}\equiv[A'F_{K}(\xi'K',\Bar{L}_{S'})+\Bar{Q}_{S'}(1-\delta)]\xi',
\end{gather*}
and $\{\Bar{C},\Bar{L},\Bar{Q},\Bar{\Tilde{\nu}},\Bar{\Tilde{\xi}},\Bar{\tau}^1\}$ are the policy functions of the future planner. Since the future planner is constrained by \eqref{eq: aggregate old bank value}, we have
\begin{equation*}
    \Bar{V}^1_{S'}=(\Bar{\Tilde{\nu}}_{S'}+\Bar{\Tilde{\xi}}_{S'})(1+\Bar{\tau}^1_{S'})N^1_{S'},\qquad
    N^1_{S'}\equiv\sigma(X_{S'}K'-D'),
\end{equation*}
where $\Bar{V}^1$ is the policy function of the future planner. Therefore,
\begin{align*}
    (1-\tau^D)\Tilde{\nu}&=(1+\Tilde{\lambda})\mathbb{E}_{z}\left[\beta\frac{U_C(\Bar{C}_{S'},\Bar{L}_{S'})}{U_C(C,L)}\left(1-\sigma+\sigma\frac{\Bar{V}^1_{S'}}{N^1_{S'}}\right)\right]R_{S},\\
    \theta\Tilde{\lambda}+(1+\tau^K)(\kappa\Tilde{\xi}+\Tilde{\nu})&=(1+\Tilde{\lambda})\mathbb{E}_{z}\left[\beta\frac{U_C(\Bar{C}_{S'},\Bar{L}_{S'})}{U_C(C,L)}\left(1-\sigma+\sigma\frac{\Bar{V}^1_{S'}}{N^1_{S'}}\right)\frac{X_{S'}}{Q_{S}}\right].
\end{align*}
The conditional expectations are entirely the functions of allocations and can be replaced by $M^D_{S}$ and $M^K_{S}$ as defined in \eqref{eq: marginal cost of deposits} and \eqref{eq: marginal benefit of capital}. It follows that $(\zeta^0,\kappa,\tau^D,\tau^K,\tau^0,\tau^1,\Tilde{\lambda},\Tilde{\nu},\Tilde{\xi})$ can be constructed using the same recipe as for the Ramsey problem. Note that \eqref{eq: aggregate bank value} allows to solve for $V$ explicitly, since $\Bar{V}^1$ is taken as given. Therefore, the problem becomes identical to that in Proposition \ref{prop: MCEA characterization}.

\end{document}